\documentclass[11pt,letterpaper]{article}

\usepackage[margin=1in]{geometry}

\usepackage{setspace}
\usepackage{parskip}
\usepackage{microtype}

\usepackage{natbib}

\usepackage{algorithm, algorithmicx, algpseudocode}
\usepackage{amssymb}
\usepackage{amsmath}
\usepackage{graphicx}
\usepackage{mathtools}
\usepackage{needspace} 
\usepackage{multicol}
\usepackage{listings}
\usepackage{amsthm}
\usepackage[shortlabels]{enumitem}
\usepackage[italicdiff]{physics}
\usepackage{cancel}
\usepackage[dvipsnames]{xcolor}
\usepackage{verbatim}
\usepackage{comment}
\usepackage{array}
\usepackage{tikz-cd}
\usepackage{import}
\usepackage{booktabs}
\usepackage{comment}
\usepackage{array}
\usepackage[colorlinks, linktoc=page, linkcolor=MidnightBlue, citecolor=blue]{hyperref}
\usepackage{forest}
\usepackage{float}
\usepackage{xspace}
\usepackage{aliascnt}
\usepackage{wrapfig}
\usepackage{thm-restate}
\usepackage{adjustbox}
\usepackage{bm}
\usepackage{bbm}
\usepackage{complexity}
\usepackage{makecell}
\usepackage{nicefrac}
\usetikzlibrary{calc, positioning, external, arrows.meta, patterns, decorations.pathreplacing}
\usepackage{cleveref}
\allowdisplaybreaks
\usepackage{booktabs}
\usepackage[table]{xcolor}
\usepackage{pifont}
\usepackage{thm-restate}

\newcommand{\cmark}{\ding{51}}
\newcommand{\xmark}{\ding{55}}

\definecolor{blueink}{HTML}{426C9C}
\definecolor{greenink}{HTML}{347967}
\definecolor{orangeink}{HTML}{BB6137}
\definecolor{grayink}{HTML}{65717D}

\let\cite\citep

\newfloat{protocol}{tbp}{lop}
\floatname{protocol}{Protocol}

\newtheorem{theorem}{Theorem}
\numberwithin{theorem}{section}
\newtheorem*{theorem*}{Theorem}

\newtheorem{lemma}[theorem]{Lemma}

\theoremstyle{definition}
\newtheorem{definition}[theorem]{Definition}

\def\cH{{\mathcal{H}}}

\usepackage[most]{tcolorbox}
\usepackage{xcolor}

\usepackage{caption}
\makeatletter
\renewcommand\thanks[1]{%
  \footnotemark%
  \protected@xdef\@thanks{\@thanks
    \protect\footnotetext[\the\c@footnote]{#1}}%
}
\makeatother

\title{Decoupling Corruption and Horizon in Robust Contextual Pricing}

\author{
Matteo Castiglioni\thanks{The authors are listed in alphabetical order.}\\
Politecnico di Milano\\
\texttt{matteo.castiglioni@polimi.it}
\and
Francesco Emanuele Stradi\footnotemark[1]\\
Politecnico di Milano\\
\texttt{francescoemanuele.stradi@polimi.it}
}

\date{} %

\begin{document}

\maketitle

\begin{abstract}\noindent
We study robust repeated contextual pricing, where valuations depend linearly on the features. At each round $t\in[T]$, a seller observes a context, posts a price, and receives only a possibly corrupted binary sale feedback. When the seller knows an upper bound $C$ on the number of corrupted rounds, we design an
algorithm with regret $\mathcal O(Cd+d^2\log T)$, where $d$ is the context dimension. This is the first guarantee for robust contextual pricing that separates the dependence on the corruption budget $C$ from the horizon $T$, closing the problem left open by Gupta, Guruganesh, Paes Leme, and Schneider (2025). 
In this work, we go even further by developing a randomized
algorithm that attains regret of order
$\mathcal{O}(Cd+d^2\log T+d\log^2T)$ with high probability, \emph{without knowing} the number of corrupted rounds $C$.
Finally, we show that randomized algorithms that do not know
$C$ \emph{cannot} guarantee regret of order $\mathcal{O}(C)+o(\log^2 T)$
with high probability, even in the one dimensional non-contextual case. 
\end{abstract}

\newpage
\tableofcontents
\newpage

\section{Introduction}
In dynamic pricing~\citep{kleinberg2003value}, a seller repeatedly interacts with arriving buyers,
posts prices, and observes only whether a sale occurred. 

We study the \emph{contextual} variant of this problem, where valuations depend linearly on
observable features. Formally, there is an unknown parameter $\theta^\star$, and at each round $t$, a context $u_t$ is revealed (chosen potentially adversarially). The context deterministically determines the buyer's
valuation $v_t = \langle \theta^\star, u_t \rangle$. After the seller posts a price $p_t$, the
only feedback is a binary sale indicator. The regret is measured against the clairvoyant benchmark
that knows $\theta^\star$ and posts $v_t$ in every round:
\[
    R_T
    =
    \sum_{t=1}^T
    \left(
        v_t
        -
        p_t\mathbb I\{p_t\leq v_t\}
    \right).
\]
This model captures feature-based pricing and is closely related to contextual search, where a
learner must locate an unknown linear threshold from one-bit
comparisons~\citep{cohen2020feature,lobel2018multidimensional,leme2022contextual,liu2021optimal}.
Both problems are now well understood in the noise-free setting: geometric algorithms based on
maintaining and cutting a feasible parameter set achieve nearly optimal regret~\citep{lobel2018multidimensional,leme2022contextual,liu2021optimal}, with the optimal
horizon dependence for contextual pricing settled at $\Theta(d \log\log T)$~\citep{liu2021optimal}.

These algorithms are, however, fundamentally fragile. They rely on the assumption that every sale bit faithfully reflects the buyer's true valuation. A natural robustness question is: what happens when an adversary can corrupt the feedback in up to $C$
rounds, modeling, for example, buyers who act irrationally or strategically?

Prior work on robust contextual pricing~\citep{krishnamurthy2020corrupted,gupta2025robust} has
made significant progress, but all existing bounds couple the corruption term and the horizon
\emph{multiplicatively}. The current state of the art, due to~\cite{gupta2025robust}, achieves
$\mathcal{O}(Cd\log\log T)$, and their lower bound rules out a fully additive $\mathcal{O}(C +
d\log\log T)$ guarantee. This  leaves open the following question:
\begin{quote}
    \emph{Can we design robust contextual pricing algorithms with regret
    $\mathcal{O}((C  + \log T )\cdot \poly(d))$?}
\end{quote}
We answer this question affirmatively, closing the gap left open
by~\citep{gupta2025robust,kalupahana2026optimalregretrobustpricing}. In doing so, we also show
that density-based methods --- developed for robust contextual search --- can be extended and adapted to the pricing setting
with the right additional ingredients, answering a question left open by \citep{leme2026density}.

In the second part of the paper, we go even further by developing a randomized algorithm that
does not require knowledge of $C$ and achieves regret
$\mathcal{O}((C  + \log^2 T )\cdot \poly(d))$ with high probability. This decouples the dependence on the
corruption term and the horizon even when $C$ is not known
to the seller. Furthermore, we show that randomized algorithms
that do not know $C$ cannot guarantee regret
$\mathcal{O}(C)+o(\log^2 T)$ with high probability, even in the
one-dimensional case. These results resolve the additional
question left open by~\citep{kalupahana2026optimalregretrobustpricing}
on the necessity of the $\log^2 T$ term for
high-probability regret with an additive $\mathcal{O}(C)$
dependence on corruption in non-contextual pricing.

\subsection{Our Results and Techniques}

We first present our result when the corruption budget is known to the seller,
and then consider the case in which it is unknown.

\subsubsection*{Known Corruption Budget}

Our first contribution is the design of an algorithm whose regret bound decouples the dependence on $C$ and $\log T$, when $C$ is known to the seller. Specifically, we prove the following.

\begin{theorem*}[\Cref{thm:standard-regret}]
There exists an algorithm with knowledge of $C$ that achieves regret:
\[
    R_T=\mathcal O\left(Cd+d^2\log T\right).
\]
\end{theorem*}
This is the first regret bound for robust contextual pricing in which the corruption term $Cd$ and the horizon term $d^2\log T$ appear \emph{additively}. This improves over the state of the art for $C=\omega(\log T)$. 

Our work is inspired by two works that provide partial solutions to our problem. On one hand, \cite{kalupahana2026optimalregretrobustpricing} shows that in the \emph{non-contextual} setting, the corruption and horizon \emph{can} be decoupled, 
achieving $\mathcal{O}(C + \log T)$ via a robust binary search paired with a commit 
procedure.   However, their approach is inherently one-dimensional (see \Cref{sec:related} for more details). Despite that, we will use the idea of a committing procedure. 

On the other hand, \cite{pmlr-v178-leme22a,leme2026density} develop density-based algorithms for robust 
\emph{contextual} search that handle changing contexts gracefully, achieving 
$\mathcal{O}(C + d\log T)$ regret. However, density-based methods do not directly yield a safe pricing
rule. A pricing algorithm must eventually post a price slightly \emph{below} the valuation, while the density-based methods only guarantee to post a price \emph{close} to it (see \Cref{sec:related} for more details).

To convey the main ideas, we first
present the algorithm in a more informative \emph{proximity-feedback} model, then we show how to
recover similar guarantees using only standard sale feedback.

\paragraph{Proximity Feedback}

We isolate the core mechanism in an idealized \emph{proximity-feedback} model where, in addition to the sale bit, the seller also learns whether the posted price was within $\epsilon$ of the true valuation. This
auxiliary model strips away the bookkeeping difficulty of distinguishing ``far'' and ``close'' queries,
and lets us focus on the central question: how can information gathered from \emph{different} contexts be reused to price a new context accurately and hence commit?

Our algorithm maintains two types of variables: a \emph{density} $\mu_t$ over the set of possible $\theta$ parameters, and $2C+1$ disjoint \emph{evidence sets} $\mathcal{H}_1,\ldots,\mathcal{H}_{2C+1}$,
each storing samples $(u_i, p_i)$ for which the proximity feedback confirmed $|\langle\theta^\star,
u_i\rangle - p_i| \leq \epsilon$.

When a new context $u_t$ arrives, the algorithm decides between committing and exploring.
The algorithm \emph{commits} when every evidence set can \emph{predict} $u_t$. Formally, we require that $u_t$ is approximately a linear combination of contexts already stored in that set,
so that the corresponding prices can be combined into a reliable estimate of
$\langle\theta^\star, u_t\rangle$. 
In this case, each evidence set $\mathcal{H}_j$
produces a confidence interval $I_j(u_t)$ for the true valuation, and the algorithm
posts the $(C+1)$-st largest lower endpoint across all $2C+1$ intervals. Since the
evidence sets are disjoint, at most $C$ can contain a corrupted observation. Using $2C+1$ sets, we discard the $C$ potentially corrupted samples --- regardless of whether corruption overestimates or underestimates their prediction --- and guarantee
a posted price that is both below the true valuation and close to it,
incurring negligible regret.

Otherwise, the algorithm \emph{explores} by posting the $\epsilon$-window median $m_t$ of $\mu_t$
in direction $u_t$ --- a generalization of the median that maintains a buffer of width $\epsilon$
around the cut point. Two separate potential arguments bound the number of exploration rounds.


\emph{Density Potential.} When the feedback indicates that $m_t$ is far from $v_t$, the algorithm
updates $\mu_t$ as in \cite{pmlr-v178-leme22a,leme2026density}. If the value is above the median, the upper tail
$\{\theta:\langle u_t,\theta\rangle\geq m_t+\epsilon/2\}$ is multiplied by $3/2$, the lower tail
is multiplied by $1/2$, and the $\epsilon$-window around $m_t$ is left unchanged; if the value is
below the median, the two tails are swapped.
The $\epsilon$-window median makes this update stable. The two tails have equal density mass, so
promoting one and demoting the other keeps the total mass well defined. Moreover, when $m_t$ is
already close to $v_t$, the region near $\theta^\star$ lies in the unchanged window and is not
penalized. Hence, on uncorrupted rounds, the density near $\theta^\star$ either increases by a constant
factor or does not decrease. Corrupted rounds can undo only a bounded amount of this progress,
yielding an $\mathcal O(C+d\log T)$ bound on the number of far-median exploration rounds.

\emph{Determinant Growth.} When the proximity feedback indicates $m_t$ is close to $v_t$, the pair $(u_t, m_t)$ is added to an evidence set that cannot yet predict $u_t$.
We show that each evidence set can receive at most $\mathcal{O}(d)$ such insertions via a different potential argument.
Intuitively, every context that the set fails to predict must contribute a new geometric direction, and the volume spanned by stored contexts can grow by a multiplicative factor only $\mathcal{O}(d)$ times before predicting the whole set of possible contexts. With $2C+1$ evidence sets, the total number of exploration rounds in which $m_t$ is close to $v_t$ is $\mathcal{O}((C+1)d)$.

In Figure~\ref{fig:augmented-geometric-view}, we provide a graphical representation of the possible algorithm decision. 
Together, these two arguments give a proximity-feedback regret bound of $\mathcal{O}(Cd + d\log T)$.

\begin{figure*}[!t]
    \centering
    \resizebox{\textwidth}{!}{
\begin{tikzpicture}[
    every node/.style={font=\small},
    arr/.style={-{Latex[length=2mm]},thick},
    axis/.style={->,black!35},
    level/.style={dashed,thick,black!55},
    paneltext/.style={font=\small,align=center}
]

\begin{scope}[xshift=0cm]
    \node[font=\bfseries] at (0,2.85) {Density update};

    \draw[axis] (-2.25,-2.25) -- (2.35,-2.25) node[right] {$\theta_1$};
    \draw[axis] (-2.25,-2.25) -- (-2.25,2.35) node[above] {$\theta_2$};

    \fill[gray!6] (0,0) circle (2);
    \draw[thick] (0,0) circle (2);
    \node[anchor=north east] at (-1.30,-1.42) {$\Theta$};

    \begin{scope}
        \clip (0,0) circle (2);
        \fill[gray!18] (-2.3,1.73) -- (2.3,-1.03) -- (2.3,-0.43) -- (-2.3,2.33) -- cycle;
        \fill[blue!10] (-2.3,2.33) -- (2.3,-0.43) -- (2.3,2.3) -- (-2.3,2.3) -- cycle;
    \end{scope}
    \draw[level] (-1.95,1.52) -- (1.95,-0.82);
    \draw[level] (-1.95,2.12) -- (1.95,-0.22);

    \draw[arr] (-1.60,-1.25) -- (-0.95,-0.15);
    \node[anchor=east] at (-1.68,-1.25) {$u_t$};

    \foreach \x/\y in {
        -1.45/-1.10, -1.05/-1.35, -0.65/-1.05, -0.25/-1.45,
         0.20/-1.25,  0.65/-1.40,  1.10/-1.05, -1.30/-0.55,
        -0.90/-0.70, -0.40/-0.55, 0.05/-0.65
    }{
        \fill[blue!25] (\x,\y) circle (0.035);
    }

    \foreach \x/\y in {
         -0.35/1.45, 0.10/1.60, 0.55/1.35,
         0.95/1.55, 1.35/1.15, 1.50/0.70, 0.75/0.95
    }{
        \fill[blue!75] (\x,\y) circle (0.065);
    }

    \fill[red!80] (0.65,0.95) circle (0.07);
    \node[anchor=west] at (0.80,0.98) {$\theta^\star$};

    \draw[arr,blue!70] (-0.25,-1.55) to[bend right=18] (0.90,1.35);

    \node at (0,-3.05) {$m_t$ outside the $\epsilon$-window};
\end{scope}

\begin{scope}[xshift=5.35cm]
    \node[font=\bfseries] at (0,2.85) {Evidence collection};

    \draw[axis] (-2.25,-2.25) -- (2.35,-2.25) node[right] {$\theta_1$};
    \draw[axis] (-2.25,-2.25) -- (-2.25,2.35) node[above] {$\theta_2$};

    \fill[gray!6] (0,0) circle (2);
    \draw[thick] (0,0) circle (2);
    \node[anchor=north east] at (-1.30,-1.42) {$\Theta$};

    \begin{scope}
        \clip (0,0) circle (2);
        \fill[green!18] (-2.3,1.73) -- (2.3,-1.03) -- (2.3,-0.43) -- (-2.3,2.33) -- cycle;
    \end{scope}
    \draw[level,green!45!black] (-1.95,1.52) -- (1.95,-0.82);
    \draw[level,green!45!black] (-1.95,2.12) -- (1.95,-0.22);

    \draw[arr] (-1.60,-1.25) -- (-0.95,-0.15);
    \node[anchor=east] at (-1.68,-1.25) {$u_t$};

    \foreach \x/\y in {
        -1.35/-1.05, -0.95/-1.25, -0.55/-0.90, -0.15/-1.15,
         0.30/-0.95,  0.75/-1.15, 1.20/-0.75,
        -1.10/0.05, -0.70/0.35, -0.25/0.55, 0.25/0.75,
         0.75/0.90, 1.25/0.95,
        -0.65/1.35, -0.10/1.55, 0.45/1.45
    }{
        \fill[black!55] (\x,\y) circle (0.04);
    }

    \fill[red!80] (0.35,0.65) circle (0.07);
    \node[anchor=west] at (0.50,0.68) {$\theta^\star$};

    \node at (0,-3.05)
        {store in $\mathcal H_j$ $(u_t,m_t)$, $m_t\simeq \langle\theta^\star,u_t\rangle$};
\end{scope}

\begin{scope}[xshift=10.75cm]
    \node[font=\bfseries] at (0,2.85) {Commit};

    \draw[->,black!70,thick] (-2.35,-2.25) -- (2.35,-2.25)
        node[right] {$\langle\theta^\star,u_t\rangle$};

    \draw[dashed,black!65,thick] (0,-2.40) -- (0,2.20);
    \node[below] at (0,-2.40) {$v_t$};

    \fill[red!8,rounded corners] (0.20,1.20) rectangle (1.85,2.10);
    \fill[green!12,rounded corners] (-0.85,0.20) rectangle (0.85,1.05);
    \fill[red!8,rounded corners] (-1.85,-0.85) rectangle (-0.15,0.05);

    \node[red!65!black,font=\scriptsize] at (1.05,2.32) {$\leq C$ incorrect};
    \node[black!70,font=\scriptsize] at (1.25,0.65) {correct};
    \node[red!65!black,font=\scriptsize] at (-1.05,-1.08) {$\leq C$ incorrect};

    \draw[red!65,dashed,very thick] (0.55,1.82) -- (1.55,1.82);
    \draw[red!65,thick] (0.55,1.68) -- (0.55,1.96);
    \draw[red!65,thick] (1.55,1.68) -- (1.55,1.96);

    \draw[red!65,dashed,very thick] (0.35,1.45) -- (1.30,1.45);
    \draw[red!65,thick] (0.35,1.31) -- (0.35,1.59);
    \draw[red!65,thick] (1.30,1.31) -- (1.30,1.59);

    \draw[green!50!black,very thick] (-0.58,0.62) -- (0.42,0.62);
    \draw[green!50!black,thick] (-0.58,0.46) -- (-0.58,0.78);
    \draw[green!50!black,thick] (0.42,0.46) -- (0.42,0.78);

    \draw[black!65,rounded corners,thick]
        (-0.70,0.38) rectangle (0.54,0.86);

    \draw[red!65,dashed,very thick] (-1.60,-0.30) -- (-0.45,-0.30);
    \draw[red!65,thick] (-1.60,-0.44) -- (-1.60,-0.16);
    \draw[red!65,thick] (-0.45,-0.44) -- (-0.45,-0.16);

    \draw[red!65,dashed,very thick] (-1.35,-0.65) -- (-0.25,-0.65);
    \draw[red!65,thick] (-1.35,-0.79) -- (-1.35,-0.51);
    \draw[red!65,thick] (-0.25,-0.79) -- (-0.25,-0.51);

    \draw[black!80,very thick] (-0.58,-2.40) -- (-0.58,-2.10);
    \node[below] at (-0.58,-2.40) {$p_t$};

    \foreach \x/\y in {
        0.55/1.82,
        0.35/1.45,
        -0.58/0.62,
        -1.60/-0.30,
        -1.35/-0.65
    }{
        \fill[black!70] (\x,\y) circle (0.032);
    }

    \node at (0,-3.05) {post a safe lower endpoint};
\end{scope}

\end{tikzpicture}
}
    \caption{Geometric view of the proximity-feedback algorithm. The dashed parallel lines delimit the $\epsilon$-window around the median price $m_t$ in the direction $u_t$. \emph{Left:} when
$\theta^\star$ is outside this window, the
density is updated toward the side containing $\theta^\star$. \emph{Middle:} when $\theta^\star$ is inside the
window, $m_t$ is a near-value observation and the context-price pair is stored in an evidence set.
\emph{Right:} once the $2C+1$ evidence sets explain the context, their value intervals determine a
conservative lower endpoint, which is posted as the commit price.}
    \label{fig:augmented-geometric-view}
\end{figure*}

\paragraph{From Proximity Feedback to Sale Feedback}

Without proximity feedback, the algorithm cannot observe whether the posted price was close to
or far from the valuation, and hence cannot route each round to the appropriate update. We resolve
this with an optimistic approach: in every exploration round, the algorithm \emph{simultaneously}
inserts the posted price into an evidence set \emph{and} updates the density, hedging against both
possibilities.

The density update is inherently robust to this: updating the density in a ``close'' round
simply leaves the potential unchanged, so the $\mathcal{O}(C + d\log T)$ bound on the rounds in which $m_t$ is far from $v_t$ carries over intact.

The evidence-set update requires more care. Without proximity feedback, incorrect samples ---
pairs $(u_t, m_t)$ for which $|\langle\theta^\star, u_t\rangle - m_t| > \epsilon$ --- can now enter
evidence sets for two reasons: corrupted feedback (at most $C$ times) and uncorrupted rounds in which $m_t$ is far from $v_t$
(at most $\mathcal{O}(C + d\log T)$ times, by the density argument). We absorb both sources into
an inflated ``perceived'' corruption $\bar{C} = \mathcal{O}(C + d\log T)$, and replace the
$2C+1$ evidence sets with $2\bar{C}+1$ sets. The commit and exploration analysis then goes through
unchanged, at the cost of $\mathcal O(\bar C d)=\mathcal O(Cd+d^2\log T)$ additional
exploration rounds to fill the extra sets. This yields our guarantee for the known $C$ case:
\[
    R_T = \mathcal{O}\left(Cd + d^2\log T\right).
\]

\subsubsection*{Unknown Corruption Budget}

We next consider algorithms that do not know $C$, against an adversary that decides whether to attack each round before the seller draws its price. We establish an upper bound of $\mathcal{O}(Cd+d^2\log T+d\log^2 T)$ on the regret with high probability, and complement it with a lower bound showing that an increase from $\mathcal{O}(\log T)$ to $\mathcal{O}(\log^2 T)$ is necessary to handle an unknown corruption budget.

\begin{theorem*}[\Cref{thm:unknown-regret,thm:unknown-high-confidence-lower-bound}]
There exists a randomized algorithm that, for every $\delta\in(0,1)$, 
guarantees, without knowing $C$ and with probability at least $1-\delta$:
\[
    R_T=\mathcal{O}\left(Cd+d^2\log T+d\log T\log(T/\delta)\right).
\]
Moreover, no randomized algorithm which does not know $C$ can
guarantee regret $\mathcal{O}(C)+o(\log^2 T)$ with probability at
least $1-1/T$ for all $C$, even when $d=1$.
\end{theorem*}

\paragraph{Randomized Commitment}
Without knowing $C$, we cannot choose a $C$-dependent number of
evidence sets. We instead maintain a \emph{single evidence set} $\mathcal H$
and \emph{reset} it whenever its prediction is contradicted by the observed
feedback. As in the known-corruption setting, we alternate between
exploration rounds, which collect evidence, and commit rounds,
which use the resulting prediction.

Our randomization is inspired by the non-contextual algorithm
of~\citep{kalupahana2026optimalregretrobustpricing}. During commitment,
their algorithm occasionally posts the upper endpoint of the current
price interval, with a probability that decreases with the number
of visits to that interval. 
Specifically, randomization is used to detect price underestimations caused by
past corruptions. Indeed, if the current interval lies below the true
valuation, posting its upper endpoint reveals the error through
a sale. However, this price may lose a sale when the interval
contains the true valuation. The algorithm therefore reduces
the probability of selecting it as the number of visits increases,
controlling the cost of these rounds. 

This decay relies on
the fact that visits to the same price interval repeat the \emph{same} test,
namely, posting its upper endpoint.
In our setting, instead, different contexts provide
information about different directions of the unknown parameter,
so many commitment rounds along one direction do not justify
reducing this probability along another. Moreover, after an evidence-set
reset, the same context may produce a different valuation prediction, whose
accuracy is not guaranteed by previous observations on that context.
We will therefore select the randomization probability using information
about both the contexts and the predictions from previous commit rounds.

\begin{figure*}[!t]
\centering
\resizebox{\textwidth}{!}{%
\begin{tikzpicture}[
    x=1.2cm,y=1.2cm,
    every node/.style={font=\small},
    arr/.style={-{Latex[length=2mm]},thick},
    axis/.style={->,black!35},
    level/.style={dashed,thick,black!55},
    leader/.style={black!40,thin},
    guide/.style={densely dotted,thin,black!35},
    measure/.style={blue!70,thick,
        {Bar[width=3pt]}-{Bar[width=3pt]}},
    paneltext/.style={font=\small,align=center}
]

\begin{scope}
    \node[font=\bfseries] at (0,1.95) {Prediction changes};

    \draw[fill=gray!6,thick,rotate=162.000000]
        (0,0) ellipse [x radius=1.380000,y radius=.260000];
    \draw[axis] (-1.55,0) -- (1.63,0) node[right] {$u$};
    \draw[axis] (0,-1.22) -- (0,1.10) node[above] {$-\widetilde v$};

    \node at (-1.02,.88) {$M_t$};
    \draw[leader] (-1.02,.71) -- (-1.02,.490664);

    \draw[guide] (.72,.05) -- (.72,-.98);
    \draw[arr,black!40] (0,0) -- (.72,-.24);
    \draw[arr,red!80] (0,0) -- (.72,-.83);
    \node[red!80,anchor=west] at (.83,-1.05) {$a_t$};

    \draw[guide] (0,0) -- (-.098201,-.085186);
    \draw[guide]
        (.315165,-.363315) -- (.216964,-.448501);
    \draw[measure]
        (-.098201,-.085186) -- (.216964,-.448501);
    \filldraw[fill=white,draw=blue!70,thick]
        (.315165,-.363315) circle (.035);
    \fill[black!65] (0,0) circle (.020);

    \node[paneltext] at (0,-1.70)
        {same context, different prediction};
\end{scope}

\begin{scope}[xshift=5.35cm]
    \node[font=\bfseries] at (0,1.95) {Matrix insertion};

    \fill[blue!10,rotate=155.682854]
        (0,0) ellipse [x radius=1.522517,y radius=.427097];
    \fill[gray!6,rotate=162.000000]
        (0,0) ellipse [x radius=1.380000,y radius=.260000];
    \draw[thick,rotate=155.682854]
        (0,0) ellipse [x radius=1.522517,y radius=.427097];
    \draw[level,rotate=162.000000]
        (0,0) ellipse [x radius=1.380000,y radius=.260000];

    \draw[axis] (-1.55,0) -- (1.63,0);
    \draw[axis] (0,-1.22) -- (0,1.10);

    \node at (-1.10,.99) {$\overline M_t$};
    \draw[leader] (-1.10,.80) -- (-1.10,.737842);
    \node at (1.22,.37) {$M_t$};
    \draw[leader] (1.18,.22) -- (1.12,-.206720);

    \draw[arr,red!80] (0,0) -- (.72,-.83);
    \node[red!80,anchor=west] at (.83,-1.05) {$a_t$};

    \draw[guide] (0,0) -- (-.098201,-.085186);
    \draw[guide]
        (.571181,-.658445) -- (.472980,-.743631);
    \draw[measure]
        (-.098201,-.085186) -- (.472980,-.743631);
    \filldraw[fill=white,draw=blue!70,thick]
        (.571181,-.658445) circle (.035);
    \fill[black!65] (0,0) circle (.020);

    \node[paneltext] at (0,-1.70)
        {update to $\overline M_t$};
\end{scope}

\begin{scope}[xshift=10.70cm]
    \node[font=\bfseries] at (0,1.95) {A failed check};

    \fill[red!8,rotate=155.682854]
        (0,0) ellipse [x radius=1.522517,y radius=.427097];
    \draw[fill=gray!6,thick,rotate=155.682854]
        (0,0) ellipse [x radius=1.187563,y radius=.333135];
    \draw[level,rotate=155.682854]
        (0,0) ellipse [x radius=1.522517,y radius=.427097];

    \draw[axis] (-1.55,0) -- (1.63,0);
    \draw[axis] (0,-1.22) -- (0,1.10);

    \node at (-1.10,.99) {$\overline M_t$};
    \draw[leader] (-1.10,.80) -- (-1.10,.737842);
    \node at (1.12,.37) {$M_{t+1}$};
    \draw[leader] (1.08,.22) -- (.92,-.182083);

    \draw[arr,red!80] (0,0) -- (.72,-.83);
    \node[red!80,anchor=west] at (.83,-1.05) {$a_t$};

    \draw[guide] (0,0) -- (-.098201,-.085186);
    \draw[guide]
        (.445521,-.513587) -- (.347320,-.598773);
    \draw[measure]
        (-.098201,-.085186) -- (.347320,-.598773);
    \filldraw[fill=white,draw=blue!70,thick]
        (.445521,-.513587) circle (.035);
    \fill[black!65] (0,0) circle (.020);

    \node[paneltext] at (0,-1.70)
        {reset $\mathcal H$ and contract the ellipse};
\end{scope}
\end{tikzpicture}%
}
\caption{Geometry of the matrix counter for $d=1$. The ellipse associated
with $M_t$ is $\{x:x^\top M_t^{-1}x\leq1\}$. The red arrow represents
$a_t=(u_t,-\widetilde v_t)$, while the parallel blue segment measures the
distance from the origin to the ellipse along the same ray. The ratio
of red to blue lengths gives the score under the displayed matrix,
equal to $\gamma_t$ in the left panel.
\emph{Left:} the gray arrow represents a previous prediction at the same
context; its endpoint lies on the same vertical dotted line as the
red arrow.
\emph{Middle:} insertion leads to $\overline M_t$ and lengthens the blue segment;
the light blue region highlights the expansion.
\emph{Right:} a failure resets $\mathcal H$ and contracts the ellipse to
$M_{t+1}=\alpha^2\overline M_t$, shortening the blue segment relative to the
middle panel; the light red region highlights the contraction.
Dashed contours show the preceding ellipse.}
\label{fig:unknown-matrix-geometry}
\end{figure*}

Specifically, let $\widetilde v_t$ be the valuation prediction obtained from
$\mathcal H$. On a commit round, we post a price
slightly above $\widetilde v_t$ with probability $q_t$, which we call \emph{checking probability}, and a price
at or slightly below it with probability $1-q_t$. The upper price tests whether the current prediction underestimates the true valuation, while the lower price tests whether it overestimates it. Observing a sale at the upper price or no sale at the lower price reveals an inconsistency with the current prediction, so we reset $\mathcal H$. In doing so, it is essential to keep $q_t$ small enough, since upper-price rounds may lose a sale even when the prediction is correct, potentially leading to large regret.

To choose $q_t$, we need to account for the contexts and predictions
encountered in previous commit rounds. Since these rounds provide
information along different directions, we summarize them in a positive
definite matrix $M_t$. Specifically, we first combine each context and
prediction in an augmented vector $a_t$, whose inner product with
$w^\star=(\theta^\star,1)$ is the prediction error:
\[
    a_t=\begin{pmatrix}u_t\\-\widetilde v_t\end{pmatrix},
    \qquad
    \langle w^\star,a_t\rangle=v_t-\widetilde v_t.
\]
The we select the upper-price probability $q_t$ proportionally to
$\gamma_t\log(T/\delta)$, where:
\[
    \gamma_t=\sqrt{a_t^\top M_t^{-1}a_t}.
\]
Geometrically, the matrix $M_t$ defines the ellipsoid $\{x:x^\top M_t^{-1}x\leq1\}$, while the score $\gamma_t$ is the length of $a_t$ divided
by the distance from the origin to its boundary along the same ray.
Increasing this distance reduces the score and
the probability of selecting the upper price.

Then, every commit round modifies $M_t$ by adding weight in the direction of $a_t$, through
the update:
\[
    \overline M_t=M_t+\frac{a_ta_t^\top}{\gamma_t}.
\]
Intuitively, this update reduces the checking probability
in the next commitment rounds with similar contexts and
predictions, as we become increasingly confident in their accuracy.
If no inconsistency is detected, we keep the updated matrix
by setting $M_{t+1}=\overline M_t$. Otherwise, we slightly
reduce this confidence by setting
$M_{t+1}=\alpha^2\overline M_t$, where
$1-\alpha=\Theta(1/\log(T/\delta))$.

Figure~\ref{fig:unknown-matrix-geometry} illustrates these updates.

\paragraph{Regret Analysis}
Given the results obtained for the known $C$ setting, proving the regret essentially reduces to deal with the checking probabilities $q_t$. Intuitively, the checking probabilities must balance two costs: $(i)$ selecting the
upper price too often can lose many sales even when the prediction
is correct, $(ii)$ while selecting it too rarely can leave an underestimation
undetected for many rounds. We control these costs separately.

\emph{Bounding Upper-Price Rounds.} We control the sum of the
checking probabilities through a potential analysis of the
log-determinant $\log\det (M_t)$, as in the evidence-set growth
argument. Our choice of $\gamma_t$ and of the matrix update
ensures that each insertion increases this potential by
$\log(1+\gamma_t)$, which controls $q_t$ up to a factor
$\log(T/\delta)$. Intuitively, large checking probabilities
must therefore be accompanied by progress in the potential.
We show that the update allows a net potential increase of
at most $\mathcal O(d\log T)$, while each reset undoes a small
amount of this progress through the discount. Our choice of
$\alpha$ keeps the resulting additional checking cost at
$\mathcal O(d)$ per reset. Accounting for these decreases
and using a concentration argument gives
$N_U=\mathcal O(d\log T\log(T/\delta)+dF)$ with high probability, where $F$ counts the resets of
$\mathcal H$ and $N_U$ counts upper-price rounds.

\emph{Controlling Undetected Errors.} We next show that the
checking probabilities are large enough to control the regret
incurred before an incorrect prediction is detected.
On an unattacked round, the upper price detects any meaningful price 
underestimation.
The definition of $\gamma_t$ connects the probability of
selecting this (upper) price to the prediction error. Intuitively,
a large underestimation can have a small checking probability
only if the matrix has already accumulated substantial weight
in the direction $w^\star$, which represents the prediction error.
The matrix update controls how much each error increases
this weight, while the discount reduces it whenever an
inconsistency is detected. We show that our choice of $\alpha$
makes these reductions sufficient to prevent earlier errors
from making subsequent checks too infrequent.
A concentration argument then controls how many times the price has been underestimated and hence the revenue lost by posting prices
below the valuation. Accounting also for corrupted rounds,
resets, and upper-price rounds, we get that the regret incurred in commit rounds is of order
$\mathcal O(C+F+N_U+\log(T/\delta))$ with high probability.

Finally, the density argument controls the number of resets,
and the evidence-set growth argument controls the cost of
rebuilding $\mathcal H$. Combining these bounds gives the
stated regret guarantee.

\paragraph{Lower Bound} 
We conclude the paper by showing that algorithms that do \emph{not}
know $C$ cannot guarantee regret $\mathcal O(C)+o(\log^2 T)$
with probability at least $1-1/T$ for every $C$.
This holds even for randomized algorithms in the one-dimensional
($d=1$) non-contextual setting.
To prove this, we argue by contradiction and assume that such
a guarantee holds. We then compare an uncorrupted instance
with valuation $1/2$ against a suitably chosen corrupted
instance with valuation $1$.

On the uncorrupted instance, the assumed guarantee gives
$o(\log^2 T)$ regret with probability at least $1-1/T$, since
$C=0$. As each price above $1/2$ incurs regret $1/2$, the total
number of such prices must also be $o(\log^2 T)$ on the same
event. Among $\Omega(\log T)$ disjoint time intervals, we can
therefore select one in which few such prices are posted with
constant probability. We use this interval to construct the corrupted
instance. Specifically, let $s$ be the number of rounds preceding
the selected interval, and let its length be proportional to $s$. The adversary
attacks the first $s$ rounds to reproduce the feedback of the
uncorrupted instance, and samples $s$ additional attacked rounds
uniformly within the selected interval, giving budget $C=2s$.

We compare the two executions using the same fixed learner
randomness. If the sampled attacks cover every round in the
selected interval in which the uncorrupted execution posts a
price above $1/2$, the adversary can report no sale in all
these rounds. Since prices at most $1/2$ already produce a
sale on both instances, this makes the observed feedback
identical throughout the interval. The learner is not given
$C$, so it also posts the same prices in both executions.
By choosing the interval length sufficiently large relative
to $s$, these low prices incur regret exceeding the assumed
guarantee on the instance with valuation $1$.
Averaging over learner randomness and the sampled attacks,
the regret exceeds the assumed bound with probability
greater than $1/T$, yielding a contradiction.

\subsection{Related Works and Why Existing Techniques Fail}\label{sec:related}  

\paragraph{Non-Contextual Robust Pricing and Binary Search} 

Learning to price was initiated by~\cite{kleinberg2003value}, who showed that the optimal regret in non-contextual settings scales as $\Theta(\log\log T)$. A distinctive feature of
this setting is that the optimal regret in the uncorrupted one-dimensional case scales as
$\Theta(\log\log T)$,  reflecting the threshold structure of the
problem: binary-search strategies can locate the unknown valuation while losing little
revenue. Dynamic pricing has since been extended to richer demand
models~\citep{besbes2009dynamic,javanmard2019dynamic} and to repeated auctions with strategic
buyers~\citep{amin2013learning,cesabianchi2015regret,Drutsa17}.

In the robust non-contextual setting, \cite{gupta2025robust} studies a cautious-buyer corruption
model (where the adversary can flip buy to no-buy but not the reverse) and achieves
$\mathcal{O}(C + \log T)$ regret. \citep{kalupahana2026optimalregretrobustpricing} extend this
to arbitrary corruptions, achieving $\mathcal{O}(C + \log T)$ via a robust binary search combined
with a commit procedure, and $\mathcal{O}(C + \log^2 T)$ when the corruption budget is unknown.
Extending their approach to the contextual setting is not automatic, since the one-dimensional algorithm
relies on endpoint consistency checks for a fixed scalar valuation $v$, but in the contextual
setting, the relevant value $v_t = \langle\theta^\star, u_t\rangle$ changes with every context.
In the scalar case, a single feedback bit suffices to decide the next binary search step;
in the contextual case, information from many different contexts must be aggregated before any
reliable price can be computed.

\paragraph{Contextual Pricing, Contextual Search, and Feasible-Set Methods}

Higher-dimensional variants of dynamic pricing have been studied under the names of contextual
pricing and contextual search. Early algorithms maintained a feasible set of parameters and refined it by cutting according to the observed feedback~\citep{cohen2020feature,lobel2018multidimensional,leme2022contextual}.
The optimal horizon dependence for contextual pricing was settled by~\cite{liu2021optimal},
who obtain $\mathcal{O}(d\log\log T + d\log d)$ regret and prove a matching $\Omega(d\log\log T)$
lower bound.

With corrupted feedback, feasible-set methods break down irreparably. A single corrupted bit
generates an incorrect halfspace that may remove $\theta^\star$ from the maintained region; later observations do not reveal which halfspace was wrong, and backtracking over all
possible subsets of corrupted cuts becomes combinatorial. This motivates approaches that never
permanently discard any parameter.

\paragraph{Robust Contextual Search and Density-Based Methods}
\cite{pmlr-v178-leme22a,leme2026density} replace feasible-set updates with multiplicative
density updates over $\Theta$. A corrupted observation can decrease the density weight near
$\theta^\star$, but cannot eliminate it, making the approach inherently robust. Their algorithms
achieve $\mathcal{O}(C + d\log(1/\varepsilon))$ regret for the $\varepsilon$-ball loss in the
contextual search setting. However, density-based methods do not directly yield a safe pricing
rule: a pricing algorithm must post a price \emph{below} the valuation, but if it keeps querying
near the density median (as \cite{pmlr-v178-leme22a,leme2026density} do), it alternates between prices just above and just below $v_t$. Above-valuation
prices incur constant regret per round, while the density near $\theta^\star$ remains essentially
stable since the queries are already close.
\citep{leme2026density} explicitly identifies pricing as an open problem where density updates alone are insufficient and additional ideas are needed.

\paragraph{Robust Contextual Pricing and the Challenge to Decouple $C$ and $T$ }

\citep{krishnamurthy2020corrupted} initiated the study of contextual search with adversarial
corruptions, obtaining guarantees applicable to pricing. The state of the art is due to~\citep{gupta2025robust}, who achieve $\mathcal{O}(Cd\log\log T)$
for robust contextual pricing with known corruption budget, improving the dimensional dependence of~\citep{krishnamurthy2023contextual}. They also show that a fully additive
$\mathcal{O}(C + d\log\log T)$ bound is unachievable in general: the corruption term cannot
simply be added to the optimal uncorrupted rate.  All these approaches follow
the same high-level template: a master algorithm selects among multiple base regret
minimizers, each designed for a different corruption level. This architecture inherently
couples $C$ and $T$.

In Table~\ref{tab:comparison-known-C}, we provide a comparison between our results and the ones of the state-of-the-art.
\begin{table}[!t]
\centering
\renewcommand{\arraystretch}{1.18}
\setlength{\tabcolsep}{6pt}
\caption{Comparison of corruption-robust pricing guarantees.}
\label{tab:comparison-known-C}
\begin{tabular}{lccc}
\toprule
\textbf{Work} &
\textbf{Regret} &
\textbf{Regret - Unknown C} &
\textbf{Contextual}\\
\midrule

\citep{krishnamurthy2023contextual} &
$\mathcal O\!\left(Cd^3\log^2 T\right)$ &
$\mathcal O\!\left((C+\log T)d^3\log^2 T\log(1/\delta)\right)$ &
\cmark \\

\citep{gupta2025robust} &
$\mathcal O\!\left(Cd\log\log T\right)$ &
$\mathcal O\!\left((C/\delta+\log(T/\delta))d\log\log T\right)$ &
\cmark \\

\citep{kalupahana2026optimalregretrobustpricing} &
$\mathcal O\!\left(C+\log T\right)$&
$\mathcal O\!\left(C+\log T+ \log T \log(T/\delta)\right)$&
\xmark \\

\rowcolor{gray!15}
\textbf{This work} &
$\mathcal O\!\left(Cd+d^2\log T\right)$ &
$\mathcal O\!\left(Cd+d^2\log T + d\log T \log(T/\delta)\right)$ &
\cmark \\
\bottomrule
\end{tabular}
\end{table}

\section{Contextual Pricing}
\label{sec:contextual-pricing}

In this section, we introduce the classical (uncorrupted) repeated contextual pricing problem and its variant with corruption.

\paragraph{Notation}
For $d \in \mathbb{N}$, let $\mathcal{B}_d(a,r) \subset \mathbb{R}^d$ denote the $d$-dimensional Euclidean ball of radius $r>0$ centered at $a \in \mathbb{R}^d$. We denote by $\|\cdot\|_2$ the Euclidean norm on $\mathbb{R}^d$, by $\langle \cdot,\cdot\rangle$ the associated inner product, and by $\det(A)$ the determinant of a square matrix $A=[a_1,...,a_d]$ with $a_i\in\mathbb{R}^d$ for all $i\in[d]$. The canonical basis of $\mathbb{R}^d$ is denoted by $\{e_1,\ldots,e_d\}$. Finally, given $n\in \mathbb N$, we let $[n] \coloneqq \{1,\dots,n\}$.

\subsection{Uncorrupted Contextual Pricing}
\label{sec:uncorrupted-pricing}

We study a repeated posted-price problem between a seller and a sequence of buyers over a horizon $T \in \mathbb{N}$.
At each round $t \in [T]$, a \emph{context} $u_t \in \mathcal{B}_d(0,1)$ is presented to the learner, encoding observable features of the arriving customer, the product, or market conditions. We make no statistical assumptions on the context sequence $(u_t)_{t \in [T]}$: it may be chosen adversarially. 
After observing $u_t$, the seller posts a price $p_t$. The buyer's valuation in round $t$ depends on an unknown parameter $\theta^\star \in \Theta \coloneqq \mathcal{B}_d(0,1)$.
In particular, the valuation is:
\[
    v_t \coloneqq \langle \theta^\star, u_t \rangle \geq 0,
\]
and the buyer purchases the product if and only if $p_t \leq v_t$. Crucially, the seller never observes $v_t$ directly; instead, the seller observes only the binary sale outcome:
\[
    y_t \coloneqq \mathbb{I}\{p_t \leq v_t\} \in \{0,1\}.
\]
The revenue collected by the seller in round $t$ is $p_t \, y_t$.
We evaluate the seller against the clairvoyant benchmark that knows $\theta^\star$ (and hence $v_t$) and, in every round, posts the revenue-maximizing feasible price $p_t^\star \coloneqq v_t$. The \emph{regret} of a pricing policy over the horizon $T$ is the cumulative shortfall in revenue with respect to this benchmark:
\[
    R_T
    \;\coloneqq\;
    \sum_{t=1}^T
    \Big(
        v_t - p_t \, \mathbb{I}\{p_t \leq v_t\}
    \Big).
\]
The seller's goal is to design a pricing policy guaranteeing small regret $R_T$ for every possible $\theta^\star \in \Theta$ and every context sequence $u_1,\dots,u_T \in \mathcal{B}_d(0,1)$.

\subsection{Robust Contextual Pricing}
\label{sec:robust-pricing}

We now introduce the robust variant of contextual pricing. The interaction protocol is unchanged, except that the sale feedback
may be \emph{corrupted} in a bounded number of rounds.

Formally, instead of observing the true feedback $y_t$, the seller
observes a (possibly corrupted) signal $Y_t \in \{0,1\}$. The
corruption level is parametrized by a budget $C$. We consider both
the \emph{known-$C$ setting}, where this budget is given to the
seller, and the \emph{unknown-$C$ setting}, where it is not.
There is a set of corrupted rounds $S \subseteq [T]$ with
$|S| \leq C$ such that:
\[
    Y_t = y_t \quad \text{for every } t \notin S,
\]
while for $t \in S$, the corrupted signal $Y_t$ may be chosen
arbitrarily---in particular, as a function of the full history,
the current context $u_t$, and the hidden
parameter $\theta^\star$. More precisely, the adversary may choose the corrupted rounds adaptively.
In the known-$C$ setting, it may decide whether $t \in S$ after observing $p_t$. In the unknown-$C$ setting, we require this decision to be made after observing $u_t$ but before the seller draws $p_t$, without access to the seller's current randomization.
Once a round is selected for corruption, the adversary may still choose $Y_t$ after observing the realized price.
The seller observes only $Y_t$ and has
no information on which rounds are corrupted.

The seller's regret $R_T$ is defined exactly as in \Cref{sec:uncorrupted-pricing}. The goal is to design an algorithm whose regret degrades gracefully with $C$. In particular, our main goal is to design algorithms which decouple the dependency on $T$ and $C$ in the regret, i.e., providing a $\poly(d) (C+\poly(\log T))$ bound.

\section{Contextual Pricing with Proximity Feedback}
\label{sec:augmented-feedback}
As a first step before providing our dynamic pricing procedure, throughout this section, we will assume to have a more powerful feedback than the one usually considered in our setting. Specifically, we assume that the seller also has access to ``proximity'' feedback, in addition to the feedback $Y_t$.  In particular, we get as feedback whether the price $p_t$ belongs to the $\epsilon$ interval around $v_t$, that is, 
\[x_t=\mathbb{I}\{|p_t-v_t|\leq \epsilon\}.\]
There is a set of corrupted rounds $S \subseteq [T]$ with $|S| \leq C$ such that:
\[
    X_t = x_t \quad \text{for every } t \notin S,
\]
while for $t \in S$, the corrupted signals $X_t$, $Y_t$ may be chosen arbitrarily.

\subsection{Algorithmic Approach}

In this section, we illustrate our algorithm, whose pseudocode is provided in \Cref{alg:augmented-pricing}. The algorithm takes as input the time horizon $T$ and a ``perceived corruption'' parameter $\bar C$. In this section, $\bar C=C$. However, highlighting that our guarantees are more general will be useful in the next section. Then, it computes the parameter $\epsilon=T^{-3}$.

\begin{algorithm}[!t]

\caption{Robust Contextual Pricing with Proximity Feedback}

\label{alg:augmented-pricing}

\begin{algorithmic}[1]
\Require Horizon $T$, Corruption parameter $\bar C$
\State Set $\epsilon=T^{-3}$
\State Initialize $\mu_1$ as the uniform density over $\Theta$
\label{line:init-density}
\State Initialize $2\bar C+1$ empty evidence sets $\mathcal H_1,\ldots,\mathcal H_{2 \bar C+1}\gets\emptyset$
\label{line:init-evidence}
\For{$t=1,\ldots,T$}
    \State Observe context $u_t$
    \If{every $\mathcal H_j$ satisfies Equation~\eqref{eq:explains} for $u_t$}
        \Comment{\textcolor{blue}{\textsc{Commit Check}}}
        \label{line:commit-check}
        \State \textsc{Commit}()
        \label{line:commit}
        \State \textsc{Continue}
    \EndIf
    \State Let $m_t$ be the $\epsilon$-window median of $\mu_t$ in direction $u_t$, as in Equation~\eqref{eq:window-median}
    \State Post price $p_t=m_t$
    \State Observe sale feedback $Y_t$ and proximity feedback $X_t$ \label{line:observe-feedback} \Comment{\textcolor{blue}{\textsc{Query}}}
        \label{line:query}
    \If{$X_t=1$}
        \Comment{\textcolor{blue}{\textsc{Evidence Collection}}}
        \label{line:evidence-collection}
        \State Insert $(u_t,m_t)$ into a set $\mathcal H_j$ that does not satisfy Equation~\eqref{eq:explains} 
        \label{line:insert-evidence}
        \State Set $\mu_{t+1}=\mu_t$
    \Else
        \Comment{\textcolor{blue}{\textsc{Density Update}}}
        \label{line:density-update}
           \State \textsc{Update-Density}()
    \EndIf
\EndFor
\end{algorithmic}
\end{algorithm}

The algorithm works with two sets of variables:
\begin{itemize}
\item A density $\mu_t$ over the set of possible vectors $\Theta$, which is initialized as the uniform density
(Line~\ref{line:init-density}),
\item $2C+1$ disjoint \emph{evidence sets} $\mathcal H_1,\ldots,\mathcal H_{2C+1}$ (Line~\ref{line:init-evidence}).
\end{itemize}

 Each evidence set
contains pairs $(u_i,p_i)$ collected in rounds in which the proximity feedback declared that
$p_i$ belongs to the $\epsilon$-ball centered in $\langle \theta^\star,u_i\rangle$. We use the following definition of correct evidence sample and correct evidence set. 
\begin{definition}[Correct evidence]
    An evidence sample $(u_i,p_i)$ is correct if $|\langle \theta^\star,u_i\rangle-p_i|\leq \epsilon$. An evidence set $\mathcal H_j$ is correct if all the samples $(u_i,p_i)\in \mathcal H_j$ are correct.
\end{definition}
Notice that in \Cref{alg:augmented-pricing}, all the collected samples are correct except the one collected in corrupted rounds. Despite that, this equivalence will no longer be true in the next section, where we will design an algorithm without the proximity feedback.

At the core of our algorithm lie the evidence sets, which drive the algorithm's decisions. At a
high level, each round falls into one of three cases. First, the algorithm checks whether the
evidence sets already contain enough information about the current context to determine an
approximately optimal price. If this is the case, the algorithm commits to that price.
Otherwise, the algorithm focuses on gathering useful information about the current context. If $X_t=1$, meaning that the posted price is
close to the valuation, the algorithm collects a new evidence sample, which is correct on
uncorrupted rounds. If $X_t=0$, meaning that the posted price is not close to the valuation, the algorithm updates the density $\mu_t$.


Now, we describe the procedures in more detail. 

\paragraph{Commit}

Evidence sets are crucial for our \textsc{Commit} procedure---which we provide in Algorithm~\ref{alg:commit}---since we use them to be robust and to be sure about our commitment even under corruption. Indeed, suppose that at most $\bar C$ of the collected evidence samples are not correct. The underlying idea is simple. Since the sets are disjoint, a single
non-correct sample can contaminate only one evidence set. Thus, with $2\bar C+1$ evidence sets, even after
$\bar C$ corruptions there are still at least $\bar C+1$ correct  sets. 
Hence, if every evidence set has enough information to predict the valuation, we can use a
simple robust committing strategy.

Formally, given an evidence set $\mathcal H_j$ of size $k$, we say that $\mathcal H_j$ predicts $u$ if:
\begin{equation}
\label{eq:explains}
    \min_{\lambda_1,\dots, \lambda_{k}}
    \left\{
        \epsilon\|\lambda\|_1
        +
        \left\|
            u-\sum_{(u_i,p_i)\in\mathcal H_j}\lambda_i u_i
        \right\|_{2}
    \right\}
    \leq \frac{1}{8T}.
\end{equation}
The first term is the
error coming from the previous $\epsilon$-accurate observations. The second term is the error
left in the direction that is not spanned by the observations in $\mathcal H_j$.
The definition above has the following interpretation. If $u$ can be written almost as a linear
combination of contexts already stored in $\mathcal H_j$, then the values observed on those
contexts can be combined to estimate the value in direction $u$. The term $\epsilon\|\lambda\|_1$
accounts for the fact that each stored value is only known up to error $\epsilon$. The residual term
accounts for the part of $u$ that is not predicted by the stored contexts. Notice that since $\epsilon$ is very small,
linear combinations with high coefficients $\lambda_i$  are still useful for committing with high accuracy.

When every $\mathcal H_j$ satisfies Equation~\eqref{eq:explains}, we call the \textsc{Commit} procedure.

\begin{algorithm}[!t]
\caption{\textsc{Commit}}
\label{alg:commit}
\begin{algorithmic}[1]
\State Compute the intervals $I_j(u_t)$ as in Equation~\eqref{eq:interval} \label{line:compute-intervals}
        \State Let $\underline v_t$ be the $(\bar C+1)$-st largest lower endpoint
        \label{line:lower-endpoint}
        \State Post price $p_t=\underline v_t$
        \State Set $\mu_{t+1}=\mu_t$ \label{line:post_commit}
\end{algorithmic}
\end{algorithm}
This procedure uses the solution to the inequality to estimate $\widehat v_j(u)=\sum_{(u_i,p_i)\in\mathcal H_j}\lambda_i p_i$. We define the corresponding confidence interval as:
\begin{equation}
\label{eq:interval}
    I_j(u)=
    \left[
        \widehat v_j(u)-\frac{1}{8T},
        \widehat v_j(u)+\frac{1}{8T}
    \right].
\end{equation}
Then, the algorithm aggregates the intervals to design a robust commitment strategy. Since at most $\bar C$ evidence
sets may not be correct, the algorithm discards some of them by taking the $(\bar C+1)$-st largest between all the lower
endpoints of the $I_{j}(u)$ intervals, which we denote by $\underline v_t$ (Line~\ref{line:lower-endpoint}). This is done to avoid the case where corrupted
evidence produces artificially small lower endpoints and forces the seller to post an overly
conservative price. At the same time, taking the $(\bar C+1)$-st largest lower endpoint remains sufficiently small. Indeed,
even if the $\bar C$ largest lower endpoints were corrupted and too high, the selected stays below the true valuation. Therefore, the price is both conservative and close enough to the
true valuation.
Finally, the algorithm posts $p_t=\underline v_t$
(Line~\ref{line:post_commit}).

\paragraph{Explore}
If the algorithm does not have enough information to commit, the algorithm explores following the same decision of~\citep{pmlr-v178-leme22a, leme2026density}, but exploiting the feedback differently. Formally, we need to define the $\epsilon$-window median.

\begin{definition}[$\epsilon$-window median]
Given a density $\mu_t$ and a direction $u_t$, we define the $\epsilon$-window median of $\mu_t$ in direction $u_t$ as the value $m_t$ such that:
\begin{equation}
\label{eq:window-median}
    \int_{\Theta}
        \mu_t(\theta)\mathbb I\{\langle u_t,\theta\rangle\leq m_t-\epsilon/2\}d\theta
    =
    \int_{\Theta}
        \mu_t(\theta)\mathbb I\{\langle u_t,\theta\rangle\geq m_t+\epsilon/2\}d\theta .
\end{equation}
\end{definition}
Similarly to~\citep{pmlr-v178-leme22a, leme2026density}, we notice that since all the distributions that we work with are derived from continuous
density functions, they do not have point masses. Hence, the $\epsilon$-window median is always well defined.

The seller posts $m_t$ (Line~\ref{line:query}), i.e., the $\epsilon$-window median of $\mu_t$ in direction $u_t$, and observes both the  sale feedback $Y_t$ and the proximity feedback $X_t$ (Line~\ref{line:observe-feedback}). 
If the proximity feedback says
that $m_t$ is close to the value, the round is used to collect new evidence samples
(Line~\ref{line:evidence-collection}). In that case, the algorithm stores $(u_t,m_t)$ in one
evidence set that does not satisfy Equation~\eqref{eq:explains}
(Line~\ref{line:insert-evidence}). If the proximity feedback says that $m_t$ is not close to
the value, the round is used to update the density through the \textsc{Update-Density} procedure (Line~\ref{line:density-update})---which we provide in Algorithm~\ref{alg:density}---. To do so, we need the sale feedback $Y_t$, which reveals whether the value is above or below $m_t$.
If the product is sold, the
density is shifted towards the right side (Line~\ref{line:right-update}); otherwise, it is shifted
towards the left side (Line~\ref{line:left-update}). We underline that, following~\citep{pmlr-v178-leme22a, leme2026density}, we keep the density $\mu_t$ constant inside the $\epsilon$-window, so that the density at any point inside the $\epsilon$-ball around $\theta^\star$ is guaranteed not to decrease in uncorrupted rounds.

\begin{algorithm}[!t]
\caption{\textsc{Update-Density}}
\label{alg:density}
\begin{algorithmic}[1]
\If{$Y_t=1$}
    \State Set
            \label{line:right-update}
            \[
                \mu_{t+1}(\theta)=
                \begin{cases}
                    \frac{3}{2}\mu_t(\theta),
                    & \langle u_t,\theta\rangle\geq m_t+\epsilon/2,\\
                    \mu_t(\theta),
                    & |\langle u_t,\theta\rangle-m_t|\leq \epsilon/2,\\
                    \frac{1}{2}\mu_t(\theta),
                    & \langle u_t,\theta\rangle\leq m_t-\epsilon/2
                \end{cases}
            \]
        \Else
            \State Set
            \label{line:left-update}
            \[
                \mu_{t+1}(\theta)=
                \begin{cases}
                    \frac{3}{2}\mu_t(\theta),
                    & \langle u_t,\theta\rangle\leq m_t-\epsilon/2,\\
                    \mu_t(\theta),
                    & |\langle u_t,\theta\rangle-m_t|\leq \epsilon/2,\\
                    \frac{1}{2}\mu_t(\theta),
                    & \langle u_t,\theta\rangle\geq m_t+\epsilon/2
                \end{cases}
            \]
        \EndIf 
\end{algorithmic}
\end{algorithm}

\subsection{A $\mathcal{O}(C d+ d\log T)$ Regret Bound}

In this section, we analyze Algorithm~\ref{alg:augmented-pricing}. Overall, the proof is composed of two main results: the regret is small in a committing round, and the number of non-committing rounds is small.

\subsubsection{Committing Incurs Small Regret}

We now prove the first result. We show that, whenever the algorithm commits, the posted price
does not exceed the true value and is only slightly below it. Therefore, the buyer still purchases,
and the algorithm loses only a small amount compared to the best feasible price.

We begin by analyzing the information provided by a single correct evidence set.

\begin{lemma}\label{lem:single-evidence-set}
	Fix a correct evidence set $\mathcal H_j$. If $\mathcal H_j$ satisfies Equation~\eqref{eq:explains}
	for a context $u$, then the interval $I_j(u)$ defined in Equation~\eqref{eq:interval} contains
	$\langle \theta^\star,u\rangle$.
\end{lemma}

\begin{proof}
Let $\lambda$ be the coefficients which satisfy Equation~\eqref{eq:explains}, and define the residual
$r=u-\sum_{(u_i,p_i)\in\mathcal H_j}\lambda_i u_i$. Since $\mathcal H_j$ is correct, every pair
$(u_i,p_i)\in\mathcal H_j$ satisfies $|\langle \theta^\star,u_i\rangle-p_i|\leq \epsilon$. Therefore:
\[
\begin{aligned}
\left|
\langle \theta^\star,u\rangle-\sum_{(u_i,p_i)\in\mathcal H_j}\lambda_i p_i
\right|
&=
\left|
\left\langle \theta^\star,
u-\sum_{(u_i,p_i)\in\mathcal H_j}\lambda_i u_i
\right\rangle
+
\sum_{(u_i,p_i)\in\mathcal H_j}
\lambda_i\big(\langle \theta^\star,u_i\rangle-p_i\big)
\right| \\
&\leq
\left|
\left\langle \theta^\star,
u-\sum_{(u_i,p_i)\in\mathcal H_j}\lambda_i u_i
\right\rangle
\right|
+
\sum_{(u_i,p_i)\in\mathcal H_j}
|\lambda_i|\,|\langle \theta^\star,u_i\rangle-p_i| \\
&\leq
|\langle \theta^\star,r\rangle|
+
\epsilon\|\lambda\|_1 \\
&\leq
\|r\|_{2}
+
\epsilon\|\lambda\|_1 .
\end{aligned}
\]
By Equation~\eqref{eq:explains}, the last quantity is at most $1/(8T)$. Since the interval in
Equation~\eqref{eq:interval} has radius $1/(8T)$, it contains $\langle \theta^\star,u\rangle$. This
concludes the proof.
\end{proof}

We now show that our commit is robust to incorrect evidence sets. Indeed, even in the
worst case, incorrect evidence sets can affect only $\bar C$ of the extreme lower endpoints.
Thus, taking the $(\bar C+1)$-st largest lower endpoint ensures that the committed price is still
controlled by a correct evidence set.

\begin{lemma}
\label{lem:safe-commit}
Suppose that at most $\bar C$ evidence sets are not correct and every evidence set satisfies \Cref{eq:explains} for a context $u_t$. Then, the price posted by the \textsc{Commit} procedure satisfies:
\[p_t\in \left[v_t-\frac{1}{4T},v_t\right].\]
\end{lemma}
\begin{proof}
Throughout this proof, we define a round correct when the sample $(u_t,p_t)$ is correct.
Since we are committing, by assumption, every evidence set satisfies Equation~\eqref{eq:explains}. Since the sample fed to evidence
sets are disjoint and there are at most $\bar C$ incorrect rounds, at most $\bar C$ evidence sets can contain
incorrect observations. Thus, at least $\bar C+1$ evidence sets are correct.

By Lemma~\ref{lem:single-evidence-set}, each correct evidence set returns an interval containing
$v_t$. Thus, every correct interval has lower endpoint at most $v_t$. Since at least $\bar C+1$ evidence sets are correct, at least $\bar C+1$ lower endpoints are no larger
than $v_t$. Therefore, the $(\bar C+1)$-st largest lower endpoint cannot exceed $v_t$, and so
$\underline v_t\leq v_t$. 

Moreover, every correct interval has radius $1/(8T)$. Hence, by Lemma~\ref{lem:single-evidence-set},
the center of every correct interval is within $1/(8T)$ of $v_t$. Therefore, every correct lower
endpoint belongs to $[v_t-1/(4T),v_t]$. Since at least $\bar C+1$ intervals are correct,
there are at least $\bar C+1$ lower endpoints that are at least $v_t-1/(4T)$. Therefore, the
$(\bar C+1)$-st largest lower endpoint cannot be smaller than $v_t-1/(4T)$. By the definition of
$\underline v_t$, this gives $
    \underline v_t\geq v_t-\frac{1}{4T}.
$
The algorithm posts $p_t=\underline v_t$. Since we already proved that
$\underline v_t\leq v_t$, this implies $p_t\leq v_t$, and so the product is sold. The regret of
the round is therefore:
\[
    v_t-p_t
    =
    v_t-\underline v_t \leq
    \frac{1}{4T} .
\]
This concludes the proof.
\end{proof}

Notice that as a direct consequence, we have that the regret in each commit round is at most $\frac{1}{4T}$.

In the rest of the section, we bound the number of non-committing rounds. To do that, we use two alternative arguments depending on whether the queried price is far from or close to the value. 

\subsubsection{Bounding Exploration Rounds with Far Medians}

We bound the number of rounds in which the feedback was $x_t=0$, exploiting the density update. From a technical perspective, this is a simple adaptation of the ideas of~\citep{pmlr-v178-leme22a, leme2026density}. We prove a slightly general statement, under the assumption that the density update might be performed even when $X_t=1$. While this is not the case for \Cref{alg:augmented-pricing}, this will be useful in the next section.
\begin{lemma}
\label{lem:density-rounds}
Let $\mathcal{T} \subseteq [T]$ denote a set of rounds in which the \textsc{Update-Density} procedure is invoked; this might include rounds in which $X_t=1$. Then, the number of rounds $t \in \mathcal{T}$ such that $x_t = 0$ is at most $4C+20d\log T$.
\end{lemma}

\begin{proof}
Let $ \mathcal{B}^\star=\mathcal B_d(\theta^\star,\epsilon/2)\cap \Theta$, and define the potential
$\Phi_t=\int_{\mathcal{B}^\star}\mu_t(\theta)d\theta$. Since $\Theta=\mathcal{B}_d(0,1)$ and
$\epsilon=T^{-3}$, the initial density satisfies $\Phi_1\geq 2^{-(d+1)} \epsilon^d$.

Consider an uncorrupted exploration round in which the proximity feedback declares that $m_t$ is outside
the $\epsilon$-ball. If the product is sold, then $v_t\geq m_t+\epsilon$. Thus, for every
$\theta\in \mathcal{B}^\star$, we have $\langle u_t,\theta\rangle\geq m_t+\epsilon/2$, and the density on
$\mathcal{B}^\star$ is multiplied by $3/2$. The case in which the product is not sold is symmetric. Therefore,
each uncorrupted outside round increases $\Phi_t$ by a factor $3/2$.

A corrupted round can decrease $\Phi_t$ by at most a factor $1/2$, because the multiplicative
update never uses a factor smaller than $1/2$. Since there are at most $C$ corrupted rounds and
$\Phi_t\leq1$ for every $t$, if $N$ denotes the number of uncorrupted outside rounds, then:
\[
    1
    \geq
    \Phi_1\left(\frac12\right)^C\left(\frac32\right)^N.
\]
Using $\Phi_1\geq2^{-(d+1)}\epsilon^d$, we obtain
$N\log(3/2)\leq C\log2+(d+1)\log2+d\log(1/\epsilon)$. Hence, recalling that
$\epsilon=T^{-3}$,
\[
    N
    \leq
    \frac{C\log 2+(d+1)\log 2+3d\log T}{\log(3/2)}
    \leq
    3C+20d\log T,
\]
where we used $T\geq2$. Noticing that there are at most $C$ corrupted rounds concludes the proof.
\end{proof} 

\subsubsection{Bounding Exploration Rounds with Close Medians}

It remains to bound the number of uncorrupted rounds in which the proximity feedback says that the query is close to the value. This is one of the novel technical components in our analysis.

As an intermediate result, we show that each evidence set remains small. The intuition is that an
evidence set should only store contexts that bring genuinely new information. If a new context can
already be reconstructed, up to the required tolerance, from the contexts that are currently stored,
then the evidence set already predicts it, and no insertion is made. Therefore, every insertion
means that the new context points in a direction that is still missing from the set.

This suggests that an evidence set cannot keep growing for too long. In a $d$-dimensional space,
there are only $d$ independent directions. Of course, the argument cannot rely on exact linear
independence, because the explanation test allows approximation error. A bit more formally, we track the largest determinant generated by any $d$
stored vectors, after adding small fictitious directions $\epsilon e_1,\ldots,\epsilon e_d$ to
avoid degeneracies. Geometrically, the absolute value of this determinant is the volume of the
parallelepiped spanned by those $d$ vectors. Thus, tracking the determinant is a way of tracking
how much $d$-dimensional volume the evidence set has accumulated.

Whenever a new context is inserted, it must be sufficiently far from the span already generated at
scale $\epsilon$; otherwise, it would have passed the explanation test. This forces the tracked
determinant, and hence the corresponding volume, to increase by a superconstant multiplicative factor.
Since all contexts are bounded, every such determinant is always at most a constant, while it
starts from roughly $\epsilon^d$ because of the fictitious directions. Hence only $\mathcal O(d)$
multiplicative increases are possible.

This is done in the following lemma.
\begin{lemma}
\label{lem:evidence-growth}
Fix an evidence set index $j$. At most $\mathcal{O}(d)$ pairs $(m_t, u_t)$, which do not satisfy \Cref{eq:explains} for $u_t$ and $\mathcal{H}_j$, can be added to $\mathcal{H}_j$. 
\end{lemma}

\begin{proof}

For ease of presentation, we first define the maximum absolute determinant of a finite set of
vectors. For any finite set $\mathcal S=\{v_1,\ldots,v_n\}$ of vectors in $\mathbb R^d$, define:
\[
    D(\mathcal S)
    =
    \max_{\substack{i_1,\ldots,i_d\in[n]}}
    \left|
        \det\left([v_{i_1},\ldots,v_{i_d}]\right)
    \right|,
\]
where $[v_{i_1},\ldots,v_{i_d}]$ denotes the $d\times d$ matrix whose columns are
$v_{i_1},\ldots,v_{i_d}$. Geometrically, $D(\mathcal S)$ is the largest volume of a
parallelepiped spanned by $d$ vectors in $\mathcal S$.

We use the fictitious contexts $\mathcal C=\{\epsilon e_1,\ldots,\epsilon e_d\}$, where $e_i$ is
the $i$-th canonical basis vector. For a fixed evidence set $\mathcal H_j$, define:
\[
    \mathcal S_j
    =
    \{u_i:(u_i,m_i)\in\mathcal H_j\}
    \cup
    \mathcal C,
    \qquad
    D=D(\mathcal S_j).
\]
It always holds that $D\geq\epsilon^d$, because
$\det([\epsilon e_1,\ldots,\epsilon e_d])=\epsilon^d$. Moreover, $D\leq1$ by Hadamard's
inequality, since all vectors in $\mathcal S_j$ have Euclidean norm at most $1$.

We fix an insertion of a new context $u$ in $\mathcal H_j$. This means
$\mathcal H_j$ does not satisfy \Cref{eq:explains} for $u$.

\paragraph{Step 1. An insertion forces large coefficients}
Let $b_1,\ldots,b_d$ be any basis formed by vectors in $\mathcal S_j$. Such a basis exists because
$\mathcal S_j$ contains the fictitious vectors $\epsilon e_1,\ldots,\epsilon e_d$. Write
$u=\sum_{k=1}^d a_k b_k$. We claim that:
\[
    \|a\|_1>\frac{1}{16T\epsilon}.
\]
Suppose, toward a contradiction, that $\|a\|_1\leq1/(16T\epsilon)$. Let $R\subseteq[d]$ be the
indices corresponding to real contexts and let $F\subseteq[d]$ be the indices corresponding to
fictitious contexts. For the residual:
\[
    r
    =
    u-\sum_{k\in R}a_k b_k
    =
    \sum_{k\in F}a_k b_k.
\]
Using the real contexts with coefficients $(a_k)_{k\in R}$, the left-hand side of
\Cref{eq:explains} is at most
$
    \epsilon\sum_{k\in R}|a_k|+\|r\|_{2}.
$
Additionally, Since each fictitious vector has the form $\epsilon e_i$, we have $\|b_k\|_{2}\leq\epsilon$ for every $k\in F$. Therefore, we can conclude:
\[
    \epsilon\sum_{k\in R}|a_k|+\|r\|_{2}
    \leq
    \epsilon\sum_{k\in R}|a_k|
    +
    \epsilon\sum_{k\in F}|a_k|
    =
    \epsilon\|a\|_1
    \leq
    \frac{1}{16T}
    \leq
    \frac{1}{8T}.
\]
Thus $\mathcal H_j$ would predict $u$, contradicting the insertion rule. Hence
$\|a\|_1>1/(16T\epsilon)$.

\paragraph{Step 2. Each insertion increases the maximum determinant}
Apply Step 1 to a basis $b_1,\ldots,b_d$ attaining the current maximum determinant, so that:
\[
    D
    =
    \left|
        \det([b_1,\ldots,b_d])
    \right|.
\]
Since $\|a\|_1>1/(16T\epsilon)$, there exists an index $k$ such that
$|a_k|>1/(16dT\epsilon)$. Consider the matrix obtained by replacing $b_k$ with $u$. By
linearity of the determinant in the $k$-th column:
\[
\begin{aligned}
\det([b_1,\ldots,b_{k-1},u,b_{k+1},\ldots,b_d]) &=
\det\left(\left[b_1,\ldots,b_{k-1},\sum_{\ell=1}^d a_\ell b_\ell,b_{k+1},\ldots,b_d\right]\right) \\
& =
\sum_{\ell=1}^d
a_\ell
\det([b_1,\ldots,b_{k-1},b_\ell,b_{k+1},\ldots,b_d]) \\
& =
a_k\det([b_1,\ldots,b_d]).
\end{aligned}
\]
All terms with $\ell\neq k$ vanish because the corresponding matrix contains two identical
columns. Therefore, after inserting $u$, the maximum determinant is at least
$|a_k|D$, and hence it increases by a factor larger than:
\[
    \frac{1}{16dT\epsilon}
    =
    \frac{T^2}{16d},
\]
where we used $\epsilon=T^{-3}$.

\paragraph{Step 3. Potential argument}
Let $M$ be the number of insertions into $\mathcal H_j$. Since the maximum determinant starts at
least $\epsilon^d$ and increases by a factor larger than $1/(16dT\epsilon)$ at every insertion,
after $M$ insertions, the maximum determinant is at least:
\[
    \epsilon^d
    \left(\frac{1}{16dT\epsilon}\right)^M
    =
    \left(\frac{1}{T}\right)^{3d}
    \left(\frac{T^2}{16d}\right)^M.
\]
On the other hand, by Hadamard's inequality, the maximum determinant is always at most $1$.
Therefore:
\[
    \left(\frac{1}{T}\right)^{3d}
    \left(\frac{T^2}{16d}\right)^M
    \leq
    1.
\]
If $T\geq16d$, then $T^2/(16d)\geq T$, and the previous inequality implies
$T^M\leq T^{3d}$. Hence $M\leq3d$. If instead $T<16d$, the trivial bound $M\leq T$ gives
$M<16d$. In both cases, $M=\mathcal O(d)$. This concludes the proof.
\end{proof}

The previous lemma implies an upper bound on the total number of evidence-collection rounds.

\begin{lemma}
\label{lem:evidence-rounds}
For \Cref{alg:augmented-pricing}, the number of exploration rounds in which $x_t=1$ is $\mathcal O(( C+1)d)$.
\end{lemma}
\begin{proof}
Since the algorithm has not committed, at least one evidence
set does not satisfy Equation~\eqref{eq:explains}. Hence, upon receiving $X_t=1$,  the algorithm inserts the new pair into one such
set. By Lemma~\ref{lem:evidence-growth}, each evidence set receives at most $\mathcal O(d)$ such
insertions. There are $2 C+1$ evidence sets, and therefore the total number of insertions
is $\mathcal O(( C+1)d)$. 

Noticing that $\sum_{t\in[T]}\mathbb{I}\{X_t\neq x_t\}\leq C$ concludes the proof.
\end{proof}

\subsubsection{Final Result}

We can now combine the previous bounds and prove the regret guarantee.

\begin{theorem}
\label{thm:augmented-regret}
Algorithm~\ref{alg:augmented-pricing} attains regret:
\[
    R_T=\mathcal O(Cd+d\log T).
\]
\end{theorem}

\begin{proof}
We decompose the regret according to the type of round: commit, exploration rounds where $x_t=0$, and exploration rounds with $x_t=1$.

By Lemma~\ref{lem:safe-commit}, each commit round has regret at most $1/4T$. Therefore, the total
regret of commit rounds is at most $1/4$.

Then, we notice that each exploration round has regret at most $2$ since $\Theta= \mathcal{B}_d(0,1)$ and
$\|u_t\|_2\leq 1$. Hence, to conclude the proof, we just need to bound the number of exploration rounds.

Consider the exploration rounds in which $x_t=0$.
By Lemma~\ref{lem:density-rounds}, the number of those rounds  is
$\mathcal O(C+d\log T)$.  

Consider the exploration rounds in which $x_t=1$.
Similarly, By Lemma~\ref{lem:evidence-rounds}, the number of those rounds is
$\mathcal O((C+1)d)$.

Putting the bounds together, we get:
\[
    R_T
    =
    \mathcal O(1)+\mathcal O(C+d\log T)+\mathcal O((C+1)d)
    =
    \mathcal O(Cd+d\log T).
\]
This concludes the proof.
\end{proof}

\section{From Proximity Feedback to Sale Feedback}
\label{sec:randomized-feedback}

In this section, we prove the main result of our paper. To do so, we show how to extend \Cref{thm:augmented-regret} to the setting with only sale feedback, i.e., $Y_t$.

The crucial difference between the two settings is that the algorithm no longer knows whether the posted price lies close to the true value
$v_t=\langle\theta^\star,u_t\rangle$. Hence, we are unable to determine whether the current evidence sample is correct.

We circumvent such an issue with an optimistic approach. In each exploration round, we both add the evidence sample to one of the evidence sets that do not satisfy \Cref{eq:explains} and update the density distribution $\mu_t$. Depending on the \emph{unknown} $x_t$ at least one of the two would be useful. In particular, if the median is close to the true value, it provides a useful value estimate. If the median is far from the true value, then the density update makes progress. The main downside of this approach is that we are adding additional not correct samples, which must be handled carefully.

\subsection{Algorithmic Approach}

Similarly to Algorithm~\ref{alg:augmented-pricing}, the algorithm maintains a
density $\mu_t$ over $\Theta$ and a collection of disjoint evidence sets. The first main difference is the corruption parameter:
\[
\qquad
\bar C=\left\lceil 5C+20d\log T\right\rceil,
\]
while we still  set $\epsilon=T^{-3}$.

Besides the choice of the $\bar C$ parameter, the crucial difference is that in an exploration round,
the algorithm both stores the pair $(u_t,m_t)$ and updates the density. The pseudocode is presented in \Cref{alg:standard-pricing}.

Notice that we artificially increase the corruption level $\bar C$ to take into account that we are adding inaccurate samples to the evidence set when $x_t=0$. Unlike in the proximity-feedback algorithm, an incorrect 
estimate can arise either from corrupted feedback or from an uncorrupted round in which the median is
too far from the true value.
This larger $\bar C$ upper bounds the number of evidence sets that may contain incorrect estimates.

\begin{algorithm}[!t]
\caption{Robust Contextual Pricing}
\label{alg:standard-pricing}
\begin{algorithmic}[1]
\Require Horizon $T$, corruption parameter $\bar C$
\State Set $\epsilon=T^{-3}$ 
\State Initialize $\mu_1$ as the uniform density over $\Theta$
    \label{line:std-init-density}
\State Initialize $2\bar C+1$ empty evidence sets $\mathcal H_1,\ldots,\mathcal H_{2\bar C+1}\gets\emptyset$
    \label{line:std-init-evidence}
\For{$t=1,\ldots,T$}
    \State Observe context $u_t$
    \If{every $\mathcal H_j$ satisfies Equation~\eqref{eq:explains} for $u_t$}
        \Comment{\textcolor{blue}{\textsc{Commit Check}}}
        \label{line:std-commit-check}
        \State \textsc{Commit}()
        \label{line:std-commit}
        \State \textsc{Continue}
    \EndIf
    \State Let $m_t$ be the $\epsilon$-window median of $\mu_t$ in direction $u_t$, as in Equation~\eqref{eq:window-median}
    \State Post price $p_t=m_t$ 
        \label{line:std-query}
    \State Observe sale feedback $Y_t$  \label{line:std-observe-feedback} 
        \State Insert $(u_t,m_t)$ into a set $\mathcal H_j$ that does not satisfy Equation~\eqref{eq:explains} \label{line:std-evidence-collection}\Comment{\textcolor{blue}{\textsc{Evidence Collection}}}
        \label{line:std-insert-evidence}
        \State \textsc{Update-Density}()\Comment{\textcolor{blue}{\textsc{Density Update}}}
        \label{line:std-density-update}
\EndFor
\end{algorithmic}
\end{algorithm}

\subsection{A $\mathcal{O}(Cd+d^2\log T)$ Regret Bound}

The analysis proceeds fairly similarly to the previous section.
We start by observing that \Cref{lem:single-evidence-set} still holds. However, we have the crucial difference that besides the corrupted rounds, we are adding additional pairs such that $|\langle u_i,\theta^\star\rangle - p_i| > \epsilon$. In particular, this happens in exploration rounds where  $x_t=0$.
Despite that,  \Cref{lem:safe-commit}, \Cref{lem:density-rounds}, \Cref{lem:evidence-growth} still hold.

The crucial difference is that now we cannot trivially bound the number of non correct evidence sets with the corruption $C$. Despite that, we are able to prove the following weaker bound.

\begin{lemma}
\label{lem:standard-evidence-rounds}
In \Cref{alg:standard-pricing}, the number of non-correct evidence sets is upper bounded by $ \bar C$.
\end{lemma}

\begin{proof}
We add a non correct sample to an evidence set in two cases: $i)$ at most $C$ times for corrupted rounds, and $ii)$ in exploration rounds in which $x_t=0$. 
This second term is trivially bounded by \Cref{lem:density-rounds}, which provides a $4C+20d\log T$ upper bound. This concludes the proof.
\end{proof}

Hence, we can conclude that,
\begin{itemize}
\item  Thanks to \Cref{lem:safe-commit} and \Cref{lem:standard-evidence-rounds}, the regret in each committing round is $\frac{1}{4T}$,
\item Thanks to \Cref{lem:density-rounds}, the number of rounds in which the (not observed) $x_t=0$ is $\mathcal{O}(C+d\log T)$,
\item Thanks to \Cref{lem:evidence-growth}, an argument similar to \Cref{lem:evidence-rounds} shows that the number of rounds in which the (not observed) $x_t=1$ is $\mathcal{O}((\bar C+1)d)$,
\end{itemize}
We now combine the previous results to get the final regret bound.
\begin{theorem}
\label{thm:standard-regret}
Algorithm~\ref{alg:standard-pricing} attains regret:
\[
    R_T=\mathcal O\left(Cd+d^2\log T\right).
\]
\end{theorem}

\begin{proof}

We decompose the regret according to the type of round: commit, exploration rounds where $x_t=0$, and exploration rounds with $x_t=1$.

By Lemma~\ref{lem:standard-evidence-rounds}, at most $\bar C$ evidence sets are not correct. Hence
Lemma~\ref{lem:safe-commit} applies, and each commit round has regret at most $1/4T$. Therefore, the total
regret of commit rounds is at most $1/4$.

Then, we notice that each exploration round has regret at most $2$ since $\Theta= \mathcal{B}_d(0,1)$ and
$\|u_t\|_2\leq 1$. Hence, to conclude the proof, we just need to bound the number of exploration rounds.

Consider the exploration rounds in which $x_t=0$.
By Lemma~\ref{lem:density-rounds}, the number of those rounds is
$\mathcal O(C+d\log T)$. 

Consider the exploration rounds in which $x_t=1$. Thanks to \Cref{lem:evidence-growth}, an argument similar to \Cref{lem:evidence-rounds} shows that the number of those rounds is
$\mathcal O((\bar C+1)d)$.

Putting the bounds together, we get:
\[
    R_T
    =
    \mathcal O(1)+\mathcal O(C+d\log T)+\mathcal O((\bar C+1)d)
    =
    \mathcal O(Cd+d^2\log T).
\]
This concludes the proof.
\end{proof}

\paragraph{Time Complexity} We conclude the section by discussing the per-round time complexity of Algorithm~\ref{alg:standard-pricing}. Specifically, the running time in each round is dominated by the computation of $i)$ the $\epsilon$-median and $ii)$ Equation~\eqref{eq:explains}. For point $i)$, we refer to the time complexity discussion of~\citep{pmlr-v178-leme22a,leme2026density}, where the authors show that the $\epsilon$-median can be computed in runtime $\mathcal{O}(T^d \cdot\poly(d,T))$. As concerns point $ii)$, we notice the prediction test in Equation~\eqref{eq:explains} is a
convex program. Indeed, if an evidence set contains $k$ entries and
$U=[u_1,\ldots,u_k]$, the test is equivalent to minimizing
$
    \epsilon\|\lambda\|_1+\|u-U\lambda\|_2
$
over $\lambda\in\mathbb R^k$. Since \emph{each} evidence set contains at most $\mathcal O(d)$ entries,
this convex program has dimension $\mathcal O(d)$ and can be solved in
$\poly(d, T)$ time to accuracy $\poly(1/T)$. Thus, the main bottleneck is the density-based procedure of~\citep{pmlr-v178-leme22a,leme2026density}, while excluding the computation
of the $\epsilon$-window median, each round requires
$\mathcal O(\bar C\,\poly(d,T))=\mathcal O(\poly(d,T))$ time.

\section{Handling Unknown $C$}
\label{sec:unknown-corruption}

In this section, we show how to remove the knowledge of $C$ from the input of the algorithm. 

The template of our algorithm, which we provide in Algorithm~\ref{alg:unknown-pricing}, follows the one of the previous section. The main difference lies in the commit phase, and in particular in the functioning of the evidence sets. Indeed, since $C$ is not known in advance, it is not possible to have a $C$-dependent number of evidence sets. We overcome this issue by employing a single evidence set which is reset any time a corruption is detected. To detect such corruptions while keeping regret small, we randomize the prices posted during the commit phase.

We note that randomization is a key tool for handling an unknown corruption budget $C$ in non-contextual robust pricing~\cite{kalupahana2026optimalregretrobustpricing}. In particular, during their commitment phase---namely, when the algorithm reaches an interval of prices $\ell=[L,R)$ of length at most $1/T$---the learner can explore the upper endpoint $R$ with a decaying probability of order $\frac{\log(T/\delta)}{s_\ell}$, where $s_\ell$ is the number of commitment rounds performed on $\ell$. This allows the learner to detect, with sufficiently high probability, whether its estimate is too low due to corrupted rounds, while keeping the expected number of exploration rounds bounded by $\mathcal{O}(\log(T) \log(T/\delta))$.
However, our contextual setting does not allow us to use the same randomization. Indeed, we need to address the following challenges:
\begin{itemize}
\item[(i)] On the one hand, whenever an inconsistency is detected, we need to reset the evidence set,  since we do not know which observations are corrupted and retaining them could lead to systematically incorrect predictions.

\item[(ii)] On the other hand, the decaying probability \emph{cannot be reset} together with the evidence set. Indeed, restarting the decaying probability after every detected inconsistency would repeatedly bring the exploration probability back to its initial value. As a result, the algorithm would pay a new logarithmic exploration cost after every reset, leading to a $C\log T$ regret.

\item[(iii)] At the same time, the exploration probability cannot decay according only to the total number of commitment rounds. In the non-contextual setting, all commitment rounds associated with a leaf correspond to repetitions of the same test, namely, posting the same upper endpoint $R$. In the contextual setting, instead, different contexts provide information about different directions of the unknown parameter. Therefore, observing many commitment rounds along one direction does not justify reducing the exploration probability when a new or poorly explored direction is encountered.

\item[(iv)] Finally, after an evidence reset, the same context $u_t$ may
lead to a different prediction $\widehat v_t$, and hence to
different prices. The feedback collected before the reset
concerns the prices associated with the previous prediction.
Therefore, having observed the same context many times does
not necessarily justify a small exploration probability
when the prediction changes. The exploration probability
must account for both the current context and the prediction
produced by the current evidence set.
\end{itemize}

To address these issues, we maintain the evidence set separately
from the information used to determine the exploration probability.
We summarize this information in a matrix that incorporates the
contexts and predictions of previous commit rounds. This matrix
plays the role of a counter that distinguishes different directions,
allowing the exploration probability to depend on both the current
context and the current prediction.

The evidence set is reset whenever inconsistent feedback is observed,
whereas the matrix retains information from previous commit rounds.
After such an inconsistency, we multiply the matrix by a factor
smaller than one, reducing the weight of past observations and
increasing the exploration probability. This preserves information
across resets while preventing earlier commit rounds from making
exploration too infrequent. As we show below, this update allows us
to control the total exploration cost without paying a new
logarithmic cost after every reset.

This will be the scope of the next sections.

\subsection{How to Properly Randomize}

As previously highlighted, Algorithm~\ref{alg:unknown-pricing} keeps a single evidence set $\mathcal H$ and resets it whenever a
commit round returns inconsistent feedback. To choose the checking
probability---that is, the probability to explore whether $\mathcal H$ has been corrupted by posting a large price---, we maintain a matrix which records the contexts and
predictions of previous commit rounds. After a reset, we discount
this matrix instead of discarding it. Thus, resetting the evidence changes the predictions without
restarting the checking probability.

To be more specific, suppose that $\mathcal H$ satisfies Equation~\eqref{eq:explains} for $u_t$,
and choose coefficients $\lambda_t$ satisfying that equation. We predict the
valuation from the stored evidence as
follows:
\begin{equation}
\label{eq:directional-prediction}
    \widehat v_t
    =\sum_{(u_i,p_i)\in\mathcal H}\lambda_{t,i}p_i,
    \qquad
    \widetilde v_t
    =\operatorname{clip}_{[0,1]}(\widehat v_t).
\end{equation}
Notice that clipping cannot increase the prediction error because $v_t\in[0,1]$.

The matrix must distinguish predictions from different evidence sets on the
same context. We use an \emph{augmented} vector $a_t$ to express every prediction error
through one fixed unknown direction $w^\star$:
\begin{equation}
\label{eq:augmented-error}
    a_t=\begin{pmatrix}u_t\\-\widetilde v_t\end{pmatrix},
    \qquad
    w^\star=\begin{pmatrix}\theta^\star\\1\end{pmatrix},
    \qquad
    \langle w^\star,a_t\rangle=v_t-\widetilde v_t,
\end{equation}
where changing the prediction changes $a_t$ even when the context is fixed.
The matrix can therefore distinguish predictions obtained from different
evidence sets on the same context. 

We summarize previous commit vectors through a positive definite matrix
$M_t\in\mathbb R^{(d+1)\times(d+1)}$, initialized at
$M_1=[\eta^2 e_1,\ldots,\eta^2 e_{d+1}]$ with $\eta=T^{-3}$. Geometrically, this matrix
defines the ellipsoid $\{x:x^\top M_t^{-1}x\leq1\}$ in the augmented
space. Its principal semiaxes have lengths equal to the square roots
of the eigenvalues of $M_t$, where larger semiaxes stand for more
information accumulated in these directions.

\begin{algorithm}[!t]
\caption{Robust Contextual Pricing with Unknown Corruption}
\label{alg:unknown-pricing}
\begin{algorithmic}[1]
\Require Horizon $T$, confidence parameter $\delta\in(0,1)$
\State Set $\epsilon=\eta=T^{-3}$ and choose $\beta,\alpha$ as in Equation~\eqref{eq:checking-multiplier}
\State Initialize $\mu_1$ as the uniform density over $\Theta$
\State Initialize $\mathcal H\gets\emptyset$ and $M_1\gets[\eta^2 e_1,\ldots,\eta^2 e_{d+1}]$
\For{$t=1,\ldots,T$}
    \State Observe context $u_t$
    \If{$\mathcal H$ satisfies Equation~\eqref{eq:explains} for $u_t$}
        \State Compute $\widetilde v_t,a_t$ using Equations~\eqref{eq:directional-prediction} and~\eqref{eq:augmented-error}
        \State Compute $\gamma_t,q_t$ using Equation~\eqref{eq:matrix-check-probability}
        \State Post $U_t$ with probability $q_t$ and $L_t$ otherwise
        \State Observe sale feedback $Y_t$
        \State Set $\overline M_{t}\gets M_t+a_ta_t^\top/\gamma_t$
        \If{$(p_t=L_t\text{ and }Y_t=0)$ or $(p_t=U_t\text{ and }Y_t=1)$}
            \State Set $\mathcal H\gets\emptyset$ and $M_{t+1}\gets\alpha^2\overline M_{t}$
        \Else 
        \State $M_{t+1}\gets\overline{M}_t$
        \EndIf
        \State Set $\mu_{t+1}\gets\mu_t$
    \Else
         \State Let $m_t$ be the $\epsilon$-window median of $\mu_t$ in direction $u_t$, as in Equation~\eqref{eq:window-median}
        \State Post $p_t=m_t$ and observe sale feedback $Y_t$
        \State Insert $(u_t,m_t)$ into $\mathcal H$ and invoke \textsc{Update-Density()}
        \State Set $M_{t+1}\gets M_t$
    \EndIf
\EndFor
\end{algorithmic}
\end{algorithm}

We then choose the checking multiplier $\beta$ and the matrix discount parameter $\alpha$ as follows:
\begin{equation}
\label{eq:checking-multiplier}
    \beta=\left\lceil8\log\left(\frac{32T}{\delta}\right)\right\rceil,
    \qquad
    \alpha=\exp(-4/\beta). 
\end{equation}
On a commit round, the matrix determines the checking probability $q_t$
through the following updates:
\begin{equation}
\label{eq:matrix-check-probability}
    \gamma_t=\sqrt{a_t^\top M_t^{-1}a_t},
    \qquad
    q_t=\min\left\{\frac12,\beta\gamma_t\right\}.
\end{equation}
Notice that, for $a_t\neq0$, the point $a_t/\gamma_t$ lies on the boundary of
the ellipsoid. Thus, the \emph{score} $\gamma_t$ is the length of $a_t$ divided by
the distance from the origin to the boundary along the same ray.
It equals one on the boundary and is smaller inside the ellipsoid.
For a fixed vector, a longer distance to the boundary gives a smaller
score and hence a smaller checking probability.

The upper and lower prices, selected with probabilities $q_t$ and $1-q_t$,
respectively, are:
\begin{equation}
\label{eq:unknown-prices}
    L_t=\max\left\{0,\widetilde v_t-\frac1{8T}\right\},
    \qquad
    U_t=\widetilde v_t+\frac1{4T}.
\end{equation}
A sale at $U_t$ or no sale at $L_t$ returns \textsc{Fail} and resets
$\mathcal H$, while the larger margin at $U_t$
ensures that a correct evidence set satisfies $L_t\leq v_t<U_t$,
including when its prediction error equals $1/(8T)$.

\begin{figure*}[!t]
\centering
\resizebox{\textwidth}{!}{%
\begin{tikzpicture}[
    every node/.style={font=\small},
    arr/.style={-{Latex[length=2mm]},thick},
    axis/.style={->,black!35},
    level/.style={dashed,thick,black!55},
    measure/.style={blue!70,thick,
        {Bar[width=3pt]}-{Bar[width=3pt]}},
    leader/.style={black!40,thin},
    guide/.style={densely dotted,thin,black!35},
    heading/.style={font=\bfseries,anchor=base,align=center},
    toplabel/.style={anchor=base,align=center},
    paneltext/.style={font=\small,anchor=base,align=center}
]

\def\titleY{2.55}
\def\labelY{1.78}
\def\summaryY{-1.95}

\begin{scope}
    \node[heading] at (0,\titleY) {Vector and score};

    \begin{scope}[xshift=-.20cm]
        \fill[gray!6] (0,0) circle (1.2);
        \draw[thick] (0,0) circle (1.2);
        \draw[axis] (-2.35,0) -- (2.75,0);
        \draw[axis] (0,-1.45) -- (0,1.45);

        \node[toplabel] at (-.60,\labelY)
            {$M_t=[\tfrac14 e_1,\tfrac14 e_2]$};
        \draw[leader] (-.60,1.52) -- (-.60,1.039230);
        \node[toplabel,red!80] at (1.70,\labelY) {$a_t=(1,0)$};

        \draw[arr,red!80] (0,0) -- (2.4,0);

        \draw[measure] (0,-.25) -- (1.2,-.25);
        \draw[guide] (1.2,-.25) -- (1.2,0);
        \node[blue!70,anchor=base] at (.60,-.60) {$\tfrac12$};
        \filldraw[fill=white,draw=blue!70,thick]
            (1.2,0) circle (.035);
        \fill[black!65] (0,0) circle (.025);
    \end{scope}

    \node[paneltext] at (0,\summaryY)
        {$\displaystyle\gamma_t=
        \frac{\textcolor{red!80}{1}}{\textcolor{blue!70}{1/2}}=2$};
\end{scope}

\begin{scope}[xshift=5.35cm]
    \node[heading] at (0,\titleY) {Matrix insertion};
    \node[toplabel] at (0,\labelY)
        {$\displaystyle\frac{a_ta_t^\top}{\gamma_t}$};
    \draw[arr,black!45] (0,1.30) -- (0,.90);

    \fill[green!18,rounded corners=1.5pt]
        (-.79,.10) rectangle (-.09,.70);
    \node[green!50!black] at (-.44,.40) {$\tfrac12$};
    \node[black!55] at (.44,.40) {$0$};
    \node[black!55] at (-.44,-.40) {$0$};
    \node[black!55] at (.44,-.40) {$0$};
    \draw[thick] (-.83,.76) -- (-.98,.76)
        -- (-.98,-.76) -- (-.83,-.76);
    \draw[thick] (.83,.76) -- (.98,.76)
        -- (.98,-.76) -- (.83,-.76);

\end{scope}

\begin{scope}[xshift=10.70cm]
    \node[heading] at (0,\titleY) {Updated ellipse};

    \begin{scope}[xshift=-.20cm]
        \fill[blue!10]
            (0,0) ellipse [x radius=2.078461,y radius=1.2];
        \fill[gray!6] (0,0) circle (1.2);
        \draw[thick]
            (0,0) ellipse [x radius=2.078461,y radius=1.2];
        \draw[level] (0,0) circle (1.2);
        \draw[axis] (-2.35,0) -- (2.75,0);
        \draw[axis] (0,-1.45) -- (0,1.45);

        \node[toplabel] at (-.60,\labelY)
            {$\overline M_t=[\tfrac34 e_1,\tfrac14 e_2]$};
        \draw[leader] (-.60,1.52) -- (-.60,1.148913);
        \node[toplabel,red!80] at (1.70,\labelY) {$a_t=(1,0)$};

        \draw[arr,red!80] (0,0) -- (2.4,0);

        \draw[measure] (0,-.25) -- (2.078461,-.25);
        \draw[guide] (2.078461,-.25) -- (2.078461,0);
        \node[blue!70,anchor=base] at (1.45,-.60)
            {$\tfrac{\sqrt3}{2}$};
        \draw[measure] (-.25,0) -- (-.25,1.2);
        \draw[guide] (-.25,1.2) -- (0,1.2);
        \node[blue!70] at (-.56,.60) {$\tfrac12$};
        \filldraw[fill=white,draw=blue!70,thick]
            (2.078461,0) circle (.035);
        \fill[black!65] (0,0) circle (.025);
    \end{scope}

\end{scope}
\end{tikzpicture}%
}
\caption{Geometric view of the matrix insertion for
$M_t=[\tfrac14 e_1,\tfrac14 e_2]$ and $a_t=(1,0)$, where $e_1,e_2$
are the standard basis vectors of $\mathbb R^2$.
\emph{Left:} $M_t$ defines a circle of radius $1/2$. The red vector
has length $1$, so its score is $\gamma_t=1/(1/2)=2$.
\emph{Middle:} $a_ta_t^\top/\gamma_t$ adds $1/2$ to the squared
horizontal semiaxis, while the vertical semiaxis is unchanged.
\emph{Right:} $\overline M_t=[\tfrac34 e_1,\tfrac14 e_2]$ defines an
ellipse with semiaxes $\sqrt3/2$ and $1/2$. The dashed circle shows
the previous boundary. The vector remains fixed, while its score
decreases to $2/\sqrt3$.}
\label{fig:normalized-insertion-geometry}
\end{figure*}

Then, every commit round inserts its vector into the matrix using the following
normalized update:
\begin{equation}
\label{eq:matrix-insertion}
    \overline M_t=M_t+\frac{a_ta_t^\top}{\gamma_t}.
\end{equation}
For $a_t\neq0$, the added term is a positive semidefinite matrix of
rank one which is zero on vectors orthogonal to $a_t$ and adds weight
in the direction of $a_t$. Consequently, the ellipsoid associated
with $\overline M_t$ contains the previous ellipsoid, and the score
of every fixed vector can only decrease. 

The division by $\gamma_t$ determines instead the weight assigned
to the augmented vector $a_t$ in the update of $\overline M_t$.
To understand its effect, consider consecutive commit
rounds with the same nonzero vector $a_t$ and no resets. Since $a_{t+1}=a_t$ and $M_{t+1}=\overline M_t$,
the scores satisfy
$\gamma_{t+1}=\gamma_t/\sqrt{1+\gamma_t}$, or equivalently $1/\gamma_{t+1}-1/\gamma_t
=1/(\sqrt{1+\gamma_t}+1)$. Once $\gamma_t\leq 1$, this means that each insertion adds between $1/(\sqrt2+1)$ and $1/2$ to $1/\gamma_t$. Thus, after sufficiently many insertions of the same vector,
$1/\gamma_t$ is of order $s$, where $s$ counts
the insertions, and the score $\gamma_t$ is of order $1/s$. Notice that the unnormalized update $M_t+a_ta_t^\top$
would instead give a decay of order $1/\sqrt{s}$ and too many checks.

When $a_t=0$, we set $\gamma_t=q_t=0$ and $\overline M_t=M_t$, so
the matrix update remains well defined even when the current vector vanishes.

Figure~\ref{fig:normalized-insertion-geometry}
illustrates the vector, its score, and the matrix update in a
two-dimensional example.

After a failure, we reset the evidence set and  discount the updated matrix as follows:
\begin{equation}
\label{eq:matrix-discount}
    M_{t+1}=
    \begin{cases}
        \alpha^2\overline M_t, & \text{if the round returns \textsc{Fail}},\\
        \overline M_t, & \text{otherwise}.
    \end{cases}
\end{equation}
Too see why the $\alpha^2$ discount after resets is essential for the functioning of our approach, consider an example
with $d=1$ and $u_t=\theta^\star=1$, where a single round is
sufficient to rebuild $\mathcal H$ after each reset. 
The adversary alternates sale and no-sale feedback on
these reconstruction rounds, starting with a sale.
Under the density update, the resulting window medians
remain between their first two values, both a constant
distance below the true valuation $1$.
Consequently, the exploration
and commit prices remain a constant distance below $1$,
incurring constant regret in every round.
Without discount, the matrix keeps growing across evidence-sets
resets. As we will show in the analysis, this bounds the
total number of upper-price rounds by
$\mathcal O(\log T\log(T/\delta))$ with high probability.
Each such round detects the underestimation and resets
$\mathcal H$, but the next reconstruction round produces
another incorrect prediction, at a cost of \emph{at most} one
corruption. Thus, a budget
$C=\mathcal O(\log T\log(T/\delta))$ is sufficient to
maintain this strategy throughout the horizon with high
probability, giving $R_T=\Omega(T)$.

The discount counteracts this effect
by reducing the weight of earlier commit rounds whenever an
inconsistency is detected. Our choice of $\alpha$ makes this
correction sufficient to control undetected errors while
preserving information across resets, without restarting the
checking probability, which instead would lead to $C\log T$ regret.

Geometrically, the multiplication by $\alpha^2$ contracts every semiaxis of the updated
ellipsoid by $\alpha$, preserving its principal directions. For any
fixed vector, this multiplies its score by $1/\alpha$ compared with
the score under $\overline M_t$, increasing the checking probability
whenever it is below its cap. Thus, clearing $\mathcal H$ discards
the evidence used for prediction, while discounting the matrix
retains the directional information accumulated across earlier
commit rounds, with a reduced weight.

\subsection{Bounding the Number of Upper-Price Rounds}

In this section, we bound the number of times the algorithm checks its prediction by posting $U_t$. 

Intuitively, we bound the number of checks by tracking the growth of the ellipsoid associated with $M_t$. The update in Equation~\eqref{eq:matrix-insertion} ensures that rounds with larger checking probabilities also produce larger expansions of this ellipsoid. Bounding its growth therefore limits how often the algorithm can check its predictions. As in the proof of Lemma~\ref{lem:evidence-growth}, we use a determinant to measure this growth. We underline that a failure contracts the ellipsoid, allowing the checking probabilities to increase again. Nonetheless, since each contraction has a controlled size, we can account for the additional checks through a term proportional to the number of failures.

Let $N_U$ be the number of commit rounds in which the algorithm
posts the upper price $U_t$, and let $F$ be the total number of
times Algorithm~\ref{alg:unknown-pricing} resets the evidence set.

\begin{lemma}
\label{lem:number-directional-checks}
For every
$\xi\in(0,1)$, with probability at least $1-\xi$:
\[
    N_U=\mathcal{O}\left(\beta d\log T+dF+\log(1/\xi)\right).
\]
\end{lemma}

\begin{proof}
We first bound the cumulative checking probability $\sum_{t\in[T]}q_t$ (Step 1 and 2). Specifically, in Step 1, we bound the determinant $\det (M_{T+1})$ by controlling the
eigenvalues of the final matrix $M_{T+1}$. Then, in Step 2, we add the decreases caused by the resets to the
determinant difference $\log\det M_{T+1}-\log\det M_1$, obtaining
a bound on the total increase caused by insertions.
Since each checking probability $q_t$ is bounded by
a multiple of the corresponding insertion's increase
in the log-determinant, this controls $\sum_{t\in[T]} q_t$. Finally, in Step 3, we control the number of checks $N_U$.

\paragraph{Step 1. Bounding the final determinant}
We start bouning the final determinant $M_{T+1}$.
For any
$w\in\mathbb R^{d+1}$, applying Cauchy-Schwarz gives:
\begin{align*}
    |\langle w,a_t\rangle|=|\langle M_t^{1/2}w,M_t^{-1/2}a_t\rangle|\leq
    \sqrt{w^\top M_tw}\sqrt{a_t^\top M_t^{-1}a_t}=\gamma_t\sqrt{w^\top M_tw}.
\end{align*}
For $a_t\neq0$, substituting this inequality into the matrix update in
Equation~\eqref{eq:matrix-insertion} yields:
\begin{align*}
    w^\top\overline M_tw
    =w^\top M_tw+\frac{|\langle w,a_t\rangle|^2}{\gamma_t}\leq w^\top M_tw+
        |\langle w,a_t\rangle|\sqrt{w^\top M_tw}\leq
    \left(\sqrt{w^\top M_tw}+|\langle w,a_t\rangle|\right)^2.
\end{align*}
Taking square roots gives the following bound, which also holds
when $a_t=0$ because then $\overline M_t=M_t$:
\begin{equation}
\label{eq:matrix-radius-growth}
    \sqrt{w^\top\overline M_tw}
    \leq\sqrt{w^\top M_tw}+|\langle w,a_t\rangle|.
\end{equation}
Since $\|a_t\|_2\leq2$, each insertion increases this quantity
by at most $2\|w\|_2$. Moreover, each reset multiplies it by $\alpha<1$,
while non-committing rounds leave this quantity unchanged. Starting
from $\sqrt{w^\top M_1w}=\eta\|w\|_2$ and applying
Equation~\eqref{eq:matrix-radius-growth} over at most $T$
commit insertions therefore gives:
\[
    \sqrt{w^\top M_{T+1}w}
    \leq(\eta+2T)\|w\|_2.
\]
For unit vectors $w$, this bounds $w^\top M_{T+1}w$ by
$(\eta+2T)^2$.

Now, let $\lambda_1,\ldots,\lambda_{d+1}>0$ be the eigenvalues of
$M_{T+1}$, and let $w_i$ be a unit eigenvector associated with
$\lambda_i$. Since the preceding bound holds for every vector
$w$, applying it to $w_i$ gives:
\[
    \lambda_i
    =w_i^\top M_{T+1}w_i
    \leq(\eta+2T)^2\|w_i\|_2^2
    =(\eta+2T)^2.
\]
Thus, all $d+1$ eigenvalues of $M_{T+1}$ are at most
$(\eta+2T)^2$. Taking their product yields:
\[
    \det(M_{T+1})
    =\prod_{i=1}^{d+1}\lambda_i
    \leq(\eta+2T)^{2(d+1)}.
\]
Moreover, since $M_1=[\eta^2 e_1,\ldots,\eta^2 e_{d+1}]$,, all its eigenvalues are
equal to $\eta^2$, and hence:
\[
    \det(M_1)=\eta^{2(d+1)}.
\]
\paragraph{Step 2. Bounding the total checking probability}
We now relate the checking probabilities to the growth of the
determinant. In a committing round, for $a_t\neq0$, the identity $\det(M+vv^\top)=\det(M)(1+v^\top M^{-1}v)$ for any invertible matrix $M$ and vector $v$ of matching dimension, gives:
\begin{align}
    \log\det(\overline M_t)-\log\det (M_t)=\log\left(
        1+\frac{a_t^\top M_t^{-1}a_t}{\gamma_t}
    \right)=\log(1+\gamma_t),
    \label{eq:matrix-potential-increase}
\end{align}
where the second equality uses
$a_t^\top M_t^{-1}a_t=\gamma_t^2$. The last identity also holds
when $a_t=0$, since both sides are zero. On a reset, the
subsequent discount changes the log-determinant as follows:
\begin{align*}
    \log\det(M_{t+1})=\log\det(\alpha^2\overline M_t)=\log\det\overline M_t+2(d+1)\log\alpha=\log\det(\overline M_t)-\frac{8(d+1)}{\beta}.
\end{align*}
On non-committing rounds, the matrix remains unchanged.
We define $\gamma_t=0$ and $\overline M_t=M_t$ on these rounds.

Then, by telescoping over the commit rounds we get:
\begin{align*}
    \sum_{t\in [T]}\log(1+\gamma_t)
    &=\sum_{t\in [T]}
        \left(\log\det (\overline M_t)-\log\det (M_t)\right)\\*
    &=\sum_{t\in [T]}
        \left(\log\det (M_{t+1})-\log\det (M_t)\right)
        +\frac{8(d+1)F}{\beta}\\*
    &=\log\left(\frac{\det (M_{T+1})}{\det (M_1)}\right)
        +\frac{8(d+1)F}{\beta}\\*
    &\leq
        2(d+1)\log\left(\frac{\eta+2T}{\eta}\right)
        +\frac{8(d+1)F}{\beta}.
\end{align*}
Now, we use this bound to upperbound the terms $q_t$.
Using $\beta\geq1$ and $\min\{1,x\}\leq2\log(1+x)$ for $x\geq0$, we obtain the following bound on $q_t$:
\[
    q_t
    =\min\left\{\frac12,\beta\gamma_t\right\}
    \leq\beta\min\{1,\gamma_t\}
    \leq2\beta\log(1+\gamma_t).
\]
Multiplying the preceding bound on the sum of logarithms by
$2\beta$, and recalling that $\eta=T^{-3}$, yields:
\begin{align}\label{eq:total-checking-probability}
    \sum_{t\in[T]}q_t\leq
        4\beta(d+1)\log(1+2T^4)+16(d+1)F=\mathcal{O}\left(\beta d\log T+dF\right).
\end{align}

\paragraph{Step 3. Bounding the  number of checks $N_U$}
Having bounded the cumulative checking probability in Equation~\eqref{eq:total-checking-probability}, it remains to control $N_U$, i.e., the number of commit rounds in which the algorithm posts the upper price $U_t$.

Let $I_t$ be the indicator of an commit round in which the algorithm posts $U_t$, and set
$I_t=q_t=0$ on non-committing rounds, so that
$N_U=\sum_{t=1}^T I_t$. Let $\mathcal F_t$ denote the information
available immediately before drawing the price at round $t$.
Conditional on this information, $q_t$ is fixed and $I_t$ is a Bernoulli with mean $q_t$. Therefore:
\begin{align*}
    \mathbb E\left[
        \exp\left(I_t- (e-1)q_t\right)
        \,\middle|\,\mathcal F_t
    \right]
    &=
        \exp\left(-(e-1)q_t\right)
        \mathbb E\left[\exp(I_t)\mid\mathcal F_t\right]\\
    &=
        \exp\left(-(e-1)q_t\right)
        \left((1-q_t)+q_t \, e\right)\\
    &=
        \exp\left(-(e-1)q_t\right)
        \left(1+(e-1)q_t\right)\\
    &\leq
        \exp\left(-(e-1)q_t\right)
        \exp\left((e-1)q_t\right)\\
    &=1,
\end{align*}
where the last step uses $1+x\leq e^x$.

Iterating the conditional bound from round $T$ back to round $1$ gives:
\[
    \mathbb E\left[
        \exp\left(
            N_U-(e-1)\sum_{t\in [T]}q_t
        \right)
    \right]\leq1.
\]
Applying Markov's
inequality to the nonnegative exponential above gives:
\begin{align*}
    \mathbb{P}\left(
        N_U>(e-1)\sum_{t\in [T]} q_t+\log(1/\xi)
    \right)&=
    \mathbb{P}\left(
        \exp\left(
            N_U-(e-1)\sum_{t\in [T]}q_t
        \right)>\frac1\xi
    \right)\\
    &\leq
    \xi\,\mathbb E\left[
        \exp\left(
            N_U-(e-1)\sum_{t\in [T]}q_t
        \right)
    \right]
    \leq\xi.
\end{align*}
Hence, with probability at least $1-\xi$, it holds:
\[
    N_U\leq(e-1)\sum_{t=1}^Tq_t+\log(1/\xi).
\]
Combining this with Equation~\eqref{eq:total-checking-probability} we get the final bound:
\[
    N_U=\mathcal O\left(
        \beta d\log T+dF+\log(1/\xi)
    \right),
\]
which proves the lemma.
\end{proof}

\subsection{Bounding the Number of Evidence-Set Resets}

As previously shown, the matrix controls the checking probabilities,
while the density argument bounds the total number of incorrect
samples inserted into $\mathcal H$.
We call an evidence set correct when all its samples are correct,
as in the previous sections. A correct evidence set can be reset
\emph{only} on an attacked round, giving at most $C$ such resets.
Every reset of an incorrect evidence set discards at least one
incorrect sample. Since the reset clears $\mathcal H$, a new
incorrect sample must be inserted before another reset of an
incorrect set can occur. Therefore, the total number of resets
is at most $C$ plus the number of incorrect samples inserted.

Recall that $F$ denotes the total number of evidence-sets resets
performed by Algorithm~\ref{alg:unknown-pricing}.
We bound $F$ in the following lemma.
\begin{lemma}
\label{lem:unknown-resets}
It holds:
\[
    F\leq6C+20d\log T.
\]
\end{lemma}

\begin{proof}
If $\mathcal H$ is correct, Lemma~\ref{lem:single-evidence-set} and clipping
give $|v_t-\widetilde v_t|\leq1/(8T)$. Then, Equation~\eqref{eq:unknown-prices}
implies $L_t\leq v_t<U_t$, so if the feedback is not corrupted, the algorithm cannot return \textsc{Fail}.
Hence, a corrected set $\mathcal H$ is reset at most $C$ times.

Each additional reset is due to an incorrect sample in the current set $\cH$. Incorrect samples are inserted only on an attacked round or when the median is outside the $\epsilon$-ball
around $v_t$. There are at most $C$ rounds of the first type and $4C+20d\log
T$ of the second type by Lemma~\ref{lem:density-rounds}. Adding the at most
$C$ resets of correct sets gives $F\leq6C+20d\log T$.
\end{proof}

\subsection{Bounding the Regret of the Commit Rounds}

It remains to bound the regret incurred while incorrect evidence sets survive, which we call commit regret $\overline{R}$. Indeed, an incorrect evidence set can keep underestimating the
valuation without being detected if the learner posts the lower endpoint of the interval $L_t$. To detect the error in the estimate, the learner must post the upper endpoint $U_t$, revealing the error and triggering a reset.


This is done in the following lemma.

\begin{restatable}{lemma}{lemdirectional}
\label{lem:directional-checking-regret}
Let $\delta\in(0,1)$. With probability at least $1-\delta/2$,
Algorithm~\ref{alg:unknown-pricing} guarantees:
\[
    \overline R=\mathcal{O}(C+F+N_U+\beta).
\]
\end{restatable}
The main idea behind the proof is to control the regret
incurred when the prediction $\widetilde v_t$ underestimates
the true valuation $v_t$, by showing that the probabilities
$q_t$ of selecting the upper price remain large enough
to detect these errors.
To understand how the matrix $M_t$ affects detection,
consider a positive prediction error
$e_t=v_t-\widetilde v_t$. Thus, the score satisfies
$\gamma_t\geq e_t/\sqrt{(w^\star)^\top M_tw^\star}$. 
Since $\gamma_t$ determines the probability $q_t$ of
selecting the upper price, $q_t$ becomes
smaller as $\sqrt{(w^\star)^\top M_tw^\star}$ grows.
We therefore track $\sqrt{(w^\star)^\top M_tw^\star}$,
comparing how much prediction errors increase it through
matrix insertions with how much resets of $\mathcal H$
decrease it through the discount.

We now notice that each underestimation $e_t>0$
contributes to the regret and increases
$\sqrt{(w^\star)^\top M_tw^\star}$ by at most $e_t$.
These increases weaken the preceding lower bound
on $q_t$, potentially allowing further errors
to go undetected. The discount counteracts this effect.
Whenever the upper price detects an underestimation,
the resulting reset multiplies
$\sqrt{(w^\star)^\top\overline M_tw^\star}$ by $\alpha$.

We compare these increases and decreases over the
whole execution. Specifically, the lower bound on $\gamma_t$ ensures
that an underestimation is either detected frequently
or, when its detection probability is small, its
detection causes a large decrease.
With our choice of $\alpha$, we show that, with high
probability, too much accumulated underestimation
would cause resets to remove more than the initial
value $M_1$ and all matrix updates can provide.
This is impossible because
$\sqrt{(w^\star)^\top M_tw^\star}$ remains nonnegative.
This comparison bounds the regret accumulated before
the underestimations are detected
by $\mathcal O(C+F+N_U+\log(T/\delta))$, giving the same
bound on the total commit regret $\overline R$.

To keep the exposition clear and concise, we defer
the full proof, which requires additional notation and technicalities, to Appendix~\ref{App:directional-checking-regret}.

\subsection{Final Result}

In this section, we bound the total regret by considering exploration and commit
rounds separately. Each exploration round adds one sample to the
current evidence set and incurs regret at most two. Then, Lemma~\ref{lem:evidence-growth} bounds the number of samples added between resets by
$\mathcal{O}(d)$ and the exploration regret is at most
$\mathcal{O}(d(F+1))$, which we control using the bound on $F$
from Lemma~\ref{lem:unknown-resets}. For the commit rounds,
Lemma~\ref{lem:directional-checking-regret} bounds $\overline R$
in terms of $C$, $F$, and $N_U$. Substituting the bounds on
resets and  rounds in which the upper price is posted, and adding the exploration
contribution, gives the following guarantee.

\begin{theorem}
\label{thm:unknown-regret}
Let $\delta\in(0,1)$,
Algorithm~\ref{alg:unknown-pricing} guarantees, with probability
at least $1-\delta$:
\[
    R_T=\mathcal{O}\left(
        Cd+d^2\log T+d\log T\log\left(\frac{T}{\delta}\right)
    \right).
\]
\end{theorem}
\begin{proof}
Lemma~\ref{lem:evidence-growth} bounds the number of samples
inserted between consecutive resets by $\mathcal{O}(d)$.
Each exploration round inserts one sample, so
Lemma~\ref{lem:unknown-resets} bounds the regret in exploration rounds by:
\[
    \mathcal{O}(d(F+1))
    =\mathcal{O}(Cd+d^2\log T).
\]
Apply Lemma~\ref{lem:number-directional-checks} with $\xi=\delta/2$ and
intersect its event with the event in
Lemma~\ref{lem:directional-checking-regret}. By a union bound, with
probability at least $1-\delta$, both estimates hold and imply the following bound on the commit rounds:
\begin{align*}
    \overline R
    =\mathcal{O}(C+F+N_U+\beta)=\mathcal{O}(C+dF+\beta d\log T)=\mathcal{O}\left(
        Cd+d^2\log T+d\log T\log\left(\frac{T}{\delta}\right)
    \right),
\end{align*}
where we used $\beta=\mathcal{O}(\log(T/\delta))$.
Combining the bounds concludes the proof.
\end{proof}

Theorem~\ref{thm:unknown-regret} shows that we can decouple the dependence on $C$ and
the dependence on $T$ in the regret bound even when the corruption
budget is unknown, by paying an additional $\log(T)$ factor.

\paragraph{Time Complexity} We conclude the section by studying the time complexity of Algorithm~\ref{alg:unknown-pricing}.
Indeed, the main computational bottleneck of
Algorithm~\ref{alg:unknown-pricing} remains the computation of
the $\epsilon$-window median during exploration, as in
Section~\ref{sec:randomized-feedback}. The prediction test in
Equation~\eqref{eq:explains} uses the same convex program,
now applied to a single evidence set containing
$\mathcal{O}(d)$ samples.
The matrix counter adds $\mathcal{O}(d^2)$ memory and
$\mathcal{O}(d^2)$ arithmetic operations per commit round.
This is achieved by maintaining $M_t^{-1}$, updating it
through the Sherman--Morrison formula after each insertion,
and rescaling it after a failure. Thus, handling an unknown
corruption budget preserves the overall worst-case running
time of $\mathcal{O}(T^d\cdot\poly(d,T))$ per round, due to the $\epsilon$-median computation.

\section{The Price of Handling Unknown C}
\label{sec:lower-bound}

In this section, we show that adapting to an unknown corruption budget
has a cost when the regret guarantee is required to hold with high
probability. Specifically, we prove that no algorithm without knowledge
of $C$ can guarantee regret $\mathcal O(C)+o(\log^2 T)$ with probability
at least $1-1/T$, even in the non-contextual setting.
Consequently, an $\mathcal O(C+\log T)$ guarantee at this confidence
level is impossible, even for randomized algorithms.

The construction uses only $d=1$ and $u_t=1$ for every $t\in[T]$.
Moreover, all attacked rounds are selected before the interaction
begins. Thus, the lower bound applies to the adversarial model of the
previous section, even when the attack are required to be
fixed in advance.

To prove the result, we compare an uncorrupted instance with valuation
$1/2$ against a suitably chosen corrupted instance with valuation $1$.
Specifically, we argue by contradiction and assume that
an algorithm without knowledge of $C$ guarantees regret
$\mathcal O(C)+o(\log^2 T)$ with probability at least $1-1/T$
for every $C$. We then compare an uncorrupted instance with
valuation $1/2$ against a suitably chosen corrupted instance
with valuation $1$.

On the uncorrupted instance, every price above $1/2$ incurs regret
$1/2$, but these are the only prices that can distinguish the two
instances. Since the uncorrupted instance has $C=0$, the assumed
regret guarantee gives $o(\log^2 T)$ regret with probability
at least $1-1/T$. Each price above $1/2$ contributes $1/2$
to this regret, so the total number of such prices must
also be $o(\log^2 T)$ on the same event. By considering $\Omega(\log T)$ disjoint time intervals,
we can therefore select one in which few prices above $1/2$
are posted with constant probability.
We then construct the
corrupted instance so that its attacks may hide the feedback
differences between the uncorrupted and the corrupted instance at every price above $1/2$ posted in this interval.
Since the learner does not know $C$, the same feedback (fixed the learner randomness) lead to the same prices in both executions.
We show that this forces the learner to keep posting low prices
on the instance with valuation $1$ long enough to violate the
regret guarantee with probability greater than $1/T$.

This is done in the following theorem.

\begin{theorem}
\label{thm:unknown-high-confidence-lower-bound}
No algorithm without knowledge of $C$, even randomized, can
guarantee regret: 
\[ R_T = \mathcal O(C)+o(\log^2 T),\] with probability
at least $1-1/T$ for every corruption budget $C$, even when $d=1$.
\end{theorem}

\begin{proof}
We fix $d=1$ and $u_t=1$ at every round, so that the valuation
is constant throughout the interaction. Fix an arbitrary
constant $a\geq0$, and define:
\[
    K=\lceil4a+5\rceil,
    \qquad
    b=\frac{1}{32\log K\log(2K)}.
\]
We show that, for sufficiently large $T$, no algorithm without
knowledge of $C$ can guarantee the following bound on every
instance and for every corruption budget $C$:
\begin{equation}
\label{eq:lower-assumed-guarantee}
    \mathbb P\left(R_T\leq aC+b\log^2 T\right)
    \geq1-\frac1T.
\end{equation}
Suppose now, for contradiction, that such a guarantee holds.

\paragraph{The uncorrupted instance}
Consider the instance with valuation $\theta^\star=1/2$
and no attacked rounds. Every price at most $1/2$ produces
a sale, and every higher price produces no sale. Denote the
prices and regret on this instance by $p_t^{(0)}$ and
$R_T^{(0)}$, respectively. Each round with $p_t^{(0)}>1/2$
contributes exactly $1/2$ to $R_T^{(0)}$, while every other
round contributes nonnegative regret. Thus, the total number
of prices above $1/2$ is at most $2b\log^2 T$ on the event:
\[
    \mathcal G=\left\{R_T^{(0)}\leq b\log^2 T\right\}.
\]
Equation~\eqref{eq:lower-assumed-guarantee} with $C=0$
gives $\mathbb P(\mathcal G)\geq1-1/T$. We use this bound
to select an interval containing few prices above $1/2$
with sufficiently large probability.

\paragraph{Selecting the interval}
We consider consecutive intervals whose lengths increase
by a fixed factor. Specifically, we focus on $m$ consecutive intervals of the form
$\{s+1,\ldots,Ks\}$, starting with
$s=\lceil b\log^2 T\rceil$ and multiplying $s$ by $K$
after each interval.
Since $b$ is a fixed positive constant, we have
$\lceil b\log^2 T\rceil\leq\sqrt T$ for sufficiently
large $T$. Consequently, for sufficiently large $T$:
\[
    m\geq\frac{\log T}{2\log K}-1
    \geq\frac{\log T}{4\log K}.
\]
For $j=0,\ldots,m-1$, let $N_j$ denote the number of
prices above $1/2$ posted in interval $j$ of the
uncorrupted execution.
On $\mathcal G$, the intervals are disjoint, so their
total number of such prices is bounded by $2b\log^2 T$.
Substituting the definition of $b$ and using the lower
bound on $m$ gives:
\begin{align*}
    \sum_{j=0}^{m-1}N_j
    \leq2b\log^2 T
    =\frac{\log^2 T}{16\log K\log(2K)}
    \leq\frac m2\cdot\frac{\log T}{2\log(2K)}.
\end{align*}
Thus, on $\mathcal G$, at least $m/2$ intervals satisfy
$N_j\leq\log T/(2\log(2K))$. Otherwise, their total
contribution would exceed the preceding bound.
Taking expectations over learner randomness gives:
\begin{align*}
    \sum_{j=0}^{m-1}
        \mathbb P\left(N_j\leq\frac{\log T}{2\log(2K)}\right)
    =
    \mathbb E\left[
        \sum_{j=0}^{m-1}
        \mathbb I\left\{N_j\leq\frac{\log T}{2\log(2K)}\right\}
    \right]
    \geq\frac m2\,\mathbb P(\mathcal G)
    \geq\frac m2\left(1-\frac1T\right).
\end{align*}
Therefore, at least one deterministic index $j^\star$ satisfies:
\begin{equation}
\label{eq:lower-selected-scale}
    \mathbb P\left(
        N_{j^\star}\leq\frac{\log T}{2\log(2K)}
    \right)
    \geq\frac12\left(1-\frac1T\right).
\end{equation}
We fix $s=K^{j^\star}\lceil b\log^2 T\rceil$, so that
interval $j^\star$ consists of rounds $\{s+1,\ldots,Ks\}$.
This choice depends on the distribution of the uncorrupted
execution, but not on the learner's realized randomness.
We now construct a corrupted instance that uses this
interval to violate the assumed regret guarantee.

\paragraph{The corrupted instance}
Consider the instance with valuation $\theta^\star=1$
and corruption budget $C=2s$. The adversary attacks all
the first $s$ rounds and a uniformly random subset
$\mathcal A$ of $s$ rounds from $\{s+1,\ldots,Ks\}$.
The subset $\mathcal A$ is chosen before the interaction,
independently of the learner's randomness. On each attacked
round, the adversary reports a sale if and only if the posted
price is at most $1/2$. On every other round, the feedback
is truthful.

Let $R_T^{(s)}$ denote the regret on this instance.
Since $2s\leq Ks\leq T$, each realization of $\mathcal A$
defines an admissible schedule of $2s$ attacked rounds.
Conditioning on this realization leaves only learner
randomness, so Equation~\eqref{eq:lower-assumed-guarantee}
gives $\mathbb P(R_T^{(s)}>2as+b\log^2 T\mid\mathcal A)
\leq1/T$. Averaging over the sampled subset therefore gives:
\begin{align}
\label{eq:lower-corrupted-guarantee}
    \mathbb P\left(R_T^{(s)}>2as+b\log^2 T\right)
    =
    \mathbb E_{\mathcal A}\left[
        \mathbb P\left(
            R_T^{(s)}>2as+b\log^2 T
            \mid\mathcal A
        \right)
    \right]\leq\frac1T.
\end{align}

\paragraph{Comparing the two executions}
We compare the learner's executions on the uncorrupted and
corrupted instances using the same learner randomness $U$.
This does not change the distribution of either execution.
We first fix $U$ and analyze the randomness of $\mathcal A$,
and then average over $U$.

Fix a realization of $U$ for which
$N_{j^\star}\leq\log T/(2\log(2K))$, and define the set
of rounds with prices above $1/2$ in the selected interval:
\[
    D_U=\left\{t\in\left\{s+1,\ldots,Ks\right\}:p_t^{(0)}>1/2\right\}.
\]
If $D_U\subseteq\mathcal A$, the two executions receive
the same feedback through round $Ks$ and post identical
prices through that round. Indeed, the first $s$ attacks
reproduce the feedback of the uncorrupted instance.
Within the following interval, every price above $1/2$
in the uncorrupted execution is attacked in the corrupted
execution, while every price at most $1/2$ produces a sale
in both instances. Since the learner is not given $C$,
equal histories and equal learner randomness imply equal
posted prices.

We next bound the probability that all rounds in $D_U$
are attacked. Since $b>0$ is fixed and $s\geq b\log^2 T$,
for sufficiently large $T$, it holds:
\[
    N_{j^\star}\leq\frac{\log T}{2\log(2K)}\leq\frac s2.
\]
Conditional on $U$, the set $D_U$ consists of $N_{j^\star}$
fixed rounds. There are $\binom{(K-1)s}{s}$ possible choices
of $\mathcal A$, all equally likely. To contain $D_U$,
a subset must include these $N_{j^\star}$ rounds and choose
its remaining $s-N_{j^\star}$ rounds from the other
$(K-1)s-N_{j^\star}$ rounds. Consequently:
\begin{align}
    \mathbb P_{\mathcal A}
        \left(D_U\subseteq\mathcal A\mid U\right)
    =
    \frac{\binom{(K-1)s-N_{j^\star}}{s-N_{j^\star}}}
         {\binom{(K-1)s}{s}}
    =
    \prod_{i=0}^{N_{j^\star}-1}\frac{s-i}{(K-1)s-i}
    \geq(2K)^{-N_{j^\star}}
    \geq(2K)^{-\frac{\log T}{2\log(2K)}}
    =\frac1{\sqrt T}.
    \label{eq:lower-masking}
\end{align}
For the first inequality, we used $s-i\geq s/2$ and
$(K-1)s-i\leq Ks$ in every factor. When $N_{j^\star}=0$,
the empty product equals one and the bound still holds.

On the event $D_U\subseteq\mathcal A$, the corrupted
execution posts prices at most $1/2$ on
$(K-1)s-N_{j^\star}$ rounds of the selected interval.
Each such round incurs regret at least $1/2$, since
the true valuation is $1$. Regret on all other rounds
is nonnegative. Using $N_{j^\star}\leq s/2$ and
$K\geq4a+5$, we obtain:
\[
    R_T^{(s)}
    \geq\frac{(K-1)s-N_{j^\star}}{2}
    \geq\left(2a+\frac74\right)s
    >2as+b\log^2 T,
\]
where the last inequality uses $s\geq b\log^2 T$.
Thus, whenever
$N_{j^\star}\leq\log T/(2\log(2K))$, the conditional
probability of violating the regret guarantee is at least
$1/\sqrt T$. For every other realization of $U$, this
probability is nonnegative. We can therefore write:
\[
    \mathbb P\left(
        R_T^{(s)}>2as+b\log^2 T\mid U
    \right)
    \geq\frac1{\sqrt T}\,
        \mathbb I\left\{
            N_{j^\star}\leq\frac{\log T}{2\log(2K)}
        \right\}.
\]
Averaging over learner randomness and applying
Equation~\eqref{eq:lower-selected-scale} gives:
\begin{align*}
    \mathbb P\left(R_T^{(s)}>2as+b\log^2 T\right)
    &=
    \mathbb E_U\left[
        \mathbb P\left(
            R_T^{(s)}>2as+b\log^2 T\mid U
        \right)
    \right]\\
    &\geq\frac1{\sqrt T}\,
        \mathbb E_U\left[
            \mathbb I\left\{
                N_{j^\star}\leq\frac{\log T}{2\log(2K)}
            \right\}
        \right]\\
    &=\frac1{\sqrt T}\,
        \mathbb P\left(
            N_{j^\star}\leq\frac{\log T}{2\log(2K)}
        \right)\\
    &\geq\frac{1-1/T}{2\sqrt T}
    >\frac1T,
\end{align*}
for sufficiently large $T$. This contradicts
Equation~\eqref{eq:lower-corrupted-guarantee}.

Since $a$ was arbitrary and $b=b(a)>0$, any additive term
$o(\log^2 T)$ independent of $C$ would be at most
$b\log^2 T$ for sufficiently large $T$, and would therefore
imply the guarantee we have just ruled out.
This concludes the proof.
\end{proof}

Theorem~\ref{thm:unknown-high-confidence-lower-bound} shows that the
dependence of Algorithm~\ref{alg:unknown-pricing} is optimal in $T$ and $C$. We underline that the lack of knowledge of
$C$ is essential to our argument. Indeed, we compare an uncorrupted instance
with a corrupted instance, using the same learner randomness
in both executions. Since the algorithm is not given $C$,
it posts the same prices as long as the observed feedback
is identical. If $C$ were given to the algorithm, the
different budgets could induce different prices even
with identical feedback, and this comparison would no
longer be valid.

\section{Open Problems}

Several natural questions remain open. Specifically,
\begin{itemize}
\item The dependence on the dimension is likely not optimal. Our regret bound for the known $C$ setting is $\mathcal O(Cd+d^2\log T)$, and the extra $d$ factor comes from the way the
standard feedback model is reduced to the proximity one. A natural goal would be to approach
$\mathcal O(C+d\log T)$. The same reasoning applies to the regret guarantees of the unknown $C$ case.
\item It would be interesting to develop algorithms that run in time which is polynomial in the instance, namely, in $\poly(T,d)$. Since, as previously discussed, the $\mathcal{O}(T^d)$ complexity derives from the computation of the $\epsilon$-median, we believe that developing polynomial time algorithms would require departing from density-based approaches.
\end{itemize}

\appendix

\section*{AI Use Statement}
In this work, we used generative AI tools to polish the writing for clarity and readability, generate figures, assist with identifying relevant literature, and, in a limited number of cases, help turn authors' ideas into clearer lemma statements. Specifically, we used these tools to help develop a proper checking probability of selecting the upper price in the unknown-corruption setting and for the construction used in the lower bound.
All AI-assisted text and mathematical arguments were
subsequently reviewed and verified by the authors.

\section{Proof of Lemma~\ref{lem:directional-checking-regret}}
\label{App:directional-checking-regret}

We first provide a concentration result used in the proof of
Lemma~\ref{lem:directional-checking-regret}.

In the following, for each commit round $t$, let $I_t\in\{0,1\}$ be the indicator of selecting the upper price. Conditional on
the information before the price draw $\mathcal{F}_t$, $I_t$ is Bernoulli
with mean $q_t$. Define $G_t=1$ if round $t$ is unattacked
and $v_t-\widetilde v_t>1/(4T)$, and $G_t=0$ otherwise.
Instead, on rounds in which no commit is performed, set
$I_t=G_t=q_t=0$.
Finally, for every $t\in[T]$, define
$S_t=\sqrt{(w^\star)^\top M_tw^\star}+2\beta$,
where $w^\star=(\theta^\star,1)$.

The following lemma relates the realized sum
$\sum_{t\in[T]}G_tI_tS_t$ to its 
expectation $\sum_{t\in[T]}G_tq_tS_t$.

\begin{lemma}
\label{lem:weighted-check-concentration}
Let $\delta\in(0,1)$. Suppose that $S_1$ is deterministic and
$\max_{t\in[T]}S_t\leq S_1+T$.
Then, with probability at least $1-\delta/2$, it holds:
\[
    \sum_{t\in[T]}G_tI_tS_t
    \geq \frac12\sum_{t\in[T]}G_tq_tS_t
    -\left(\max_{t\in[T]}S_t+1\right)\log(2T/\delta).
\]
\end{lemma}

\begin{proof}
Let $\mathcal F_t$ contain the history and all choices
made before the price draw, including the current context
and the adversary's attack decision, which precedes the
learner's random draw. The quantities $G_t$, $S_t$, and $q_t$
are determined by this information, while $I_t$ is conditionally
Bernoulli with mean $q_t$. In particular:
\[
    \mathbb E[G_tI_tS_t\mid\mathcal F_t]=G_tq_tS_t.
\]
Thus, $\sum_{t\in[T]}G_tq_tS_t$ is the sum of the
conditional expected weights. We show that the realized
sum cannot be much smaller.

To apply a concentration argument, we first need bounded weights.
The quantity $\max_{t\in[T]}S_t$ depends on the realized execution.
We therefore first prove an inequality for a fixed deterministic
threshold $H>0$, and then make it hold simultaneously for a
collection of thresholds. For this fixed $H$, define:
\[
    z_t=G_t\frac{\min\{S_t,H\}}{H}.
\]
These weights belong to $[0,1]$ on every realization and are
determined by $\mathcal F_t$. Conditionally on this
information, we average over $I_t=0$ and $I_t=1$.
The factor $\exp(q_tz_t/2)$ is fixed under this conditioning.
Using $\exp(-z)\leq1-z/2$ for $z\in[0,1]$, followed by
$1-x\leq\exp(-x)$, gives:
\begin{align*}
    \mathbb E\left[\exp(q_tz_t/2-I_tz_t)\mid\mathcal F_t\right]
    &=\exp(q_tz_t/2)
        \left((1-q_t)+q_t\exp(-z_t)\right)\\
    &\leq\exp(q_tz_t/2)(1-q_tz_t/2)\\
    &\leq1.
\end{align*}
The first inequality follows from convexity, since
$\exp(-z)\leq1-(1-\exp(-1))z\leq1-z/2$ on $[0,1]$.

We apply this conditional bound successively over the rounds.
For every $n\in[T]$, the terms from rounds $t<n$ are determined
by $\mathcal F_n$. Taking the conditional expectation
of the last factor gives:
\begin{align*}
    \mathbb E\left[
        \exp\left(\frac12\sum_{t=1}^{n}q_tz_t-\sum_{t=1}^{n}I_tz_t\right)
    \right]&=
    \mathbb E\left[
        \exp\left(\frac12\sum_{t=1}^{n-1}q_tz_t-\sum_{t=1}^{n-1}I_tz_t\right)
        \mathbb E\left[\exp(q_nz_n/2-I_nz_n)\mid\mathcal F_n\right]
    \right]\\
    &\leq
    \mathbb E\left[
        \exp\left(\frac12\sum_{t=1}^{n-1}q_tz_t-\sum_{t=1}^{n-1}I_tz_t\right)
    \right].
\end{align*}
Repeating this argument until the sums are empty, whose
exponential equals one, yields:
\[
    \mathbb E\left[
        \exp\left(\frac12\sum_{t\in[T]}q_tz_t-\sum_{t\in[T]}I_tz_t\right)
    \right]\leq1.
\]
We now apply Markov's inequality to this nonnegative exponential.
Since the exponential function is increasing, we obtain:
\begin{align*}
    \mathbb{P}\left(
        \frac12\sum_{t\in[T]}q_tz_t-\sum_{t\in[T]}I_tz_t
        >\log(2T/\delta)
    \right)
    &=
    \mathbb{P}\left(
        \exp\left(\frac12\sum_{t\in[T]}q_tz_t-\sum_{t\in[T]}I_tz_t\right)
        >\frac{2T}{\delta}
    \right)\\
    &\leq
    \frac{\delta}{2T}
    \mathbb E\left[
        \exp\left(\frac12\sum_{t\in[T]}q_tz_t-\sum_{t\in[T]}I_tz_t\right)
    \right]
    \leq\frac{\delta}{2T}.
\end{align*}
On the complementary event, which has probability at least
$1-\delta/(2T)$, the expression inside the exponential is
at most $\log(2T/\delta)$. Rearranging gives:
\[
    \sum_{t\in[T]}I_tz_t
    \geq\frac12\sum_{t\in[T]}q_tz_t-\log(2T/\delta).
\]
Substituting the definition of $z_t$ and multiplying by $H>0$,
we obtain:
\begin{equation}
\label{eq:truncated-weighted-checks}
    \sum_{t\in[T]}G_tI_t\min\{S_t,H\}
    \geq\frac12\sum_{t\in[T]}G_tq_t\min\{S_t,H\}
        -H\log(2T/\delta).
\end{equation}
This inequality has so far been proved for a fixed threshold.
We apply it to the $T$ thresholds
$H\in\{S_1+1,\ldots,S_1+T\}$, which are deterministic for each
fixed instance, even though $S_1$ is unknown to the learner.
A union bound shows that the probability of failure for at
least one threshold is at most
$T\cdot\delta/(2T)=\delta/2$.
Hence, Equation~\eqref{eq:truncated-weighted-checks} holds
simultaneously for all these thresholds with probability
at least $1-\delta/2$.

We work on this simultaneous event.
Since $S_1\leq\max_{t\in[T]}S_t\leq S_1+T$, the grid
contains a threshold at least as large as the largest realized
weight and at distance at most one from it. We may select
this threshold after observing the execution because the
preceding event holds for every grid point simultaneously.
For this choice:
\[
    \max_{t\in[T]}S_t\leq H
    \leq\max_{t\in[T]}S_t+1.
\]
Since $S_t\leq H$ for every $t$, the truncation disappears.

Using the upper bound on $H$ and
Equation~\eqref{eq:truncated-weighted-checks} gives:
\begin{align*}
    \sum_{t\in[T]}G_tI_tS_t
    &\geq\frac12\sum_{t\in[T]}G_tq_tS_t-H\log(2T/\delta)\\
    &\geq\frac12\sum_{t\in[T]}G_tq_tS_t
        -\left(\max_{t\in[T]}S_t+1\right)\log(2T/\delta).
\end{align*}
The simultaneous event used above has probability at least
$1-\delta/2$, which concludes the proof.
\end{proof}

We are now ready to provide the full proof of Lemma~\ref{lem:directional-checking-regret}, which we restate for convenience.
\lemdirectional*
\begin{proof}
For each commit round, let $e_t=v_t-\widetilde v_t$ be the
prediction error. We set $G_t=1$ if the round is unattacked
and $e_t>1/(4T)$, and $G_t=0$ otherwise. Thus, $G_t$ identifies
the unattacked rounds on which the upper price would reveal
an underestimation. We also set $I_t=1$ when the learner selects
the upper price $U_t$ and $I_t=0$ when it selects the lower price $L_t$.
For simplicity, we set $e_t=G_t=I_t=q_t=0$ on those rounds
in which no commit is performed. 

\paragraph{Step 1. Decomposing $\overline{R}$ in prediction errors}
We first notice that large underestimations on unattacked rounds
can incur regret without being detected when the lower price
is selected. We collect these errors in $A$ and the remaining
absolute errors in $B$:
\[
    A=\sum_{t\in[T]}G_te_t,
    \qquad
    B=\sum_{t\in[T]}(1-G_t)|e_t|,
    \qquad
    A+B=\sum_{t\in[T]}|e_t|.
\]
The sum $A$ includes the errors on all rounds with $G_t=1$,
whether the learner selects the lower or the upper price.
We first bound $B$ directly, and then turn to $A$.

Non-committing rounds contribute zero to $B$. We partition
the commit rounds with $G_t=0$ into three classes:
$(i)$ attacked rounds,
$(ii)$ unattacked rounds in which $U_t$ is posted, and
$(iii)$ unattacked rounds in which $L_t$ is posted.
Since $|e_t|\leq1$, the first two classes contribute at most
$C$ and $N_U$, respectively.

For the third class, being unattacked and having $G_t=0$
implies $e_t\leq1/(4T)$. If $e_t<-1/(8T)$, then $L_t>v_t$.
Since the round is unattacked, truthful feedback reports no sale, and the algorithm returns \textsc{Fail}. There are at most
$F$ such rounds, each contributing at most one to $B$.
Every remaining round satisfies:
\[
    -\frac1{8T}\leq e_t\leq\frac1{4T},
\]
so its absolute error is at most $1/(4T)$. These rounds
therefore contribute at most $1/4$ in total. Combining the
bounds gives:
\begin{equation}
\label{eq:other-commit-errors}
    B\leq C+N_U+F+\frac14.
\end{equation}
We now bound $\overline R$. Each commit round incurs regret
at most one. Hence, rounds that are attacked, post $U_t$,
or trigger a reset contribute at most $C+N_U+F$ in total.
Every remaining commit round is unattacked, posts $L_t$,
and does not return \textsc{Fail}. A sale therefore occurs,
and the regret on such a round satisfies:
\[
    v_t-L_t
    \leq v_t-\widetilde v_t+\frac1{8T}
    =e_t+\frac1{8T}.
\]
If $G_t=0$, then $e_t\leq1/(4T)$, so this regret is at most
$3/(8T)$. If $G_t=1$, the regret is at most $e_t+1/(8T)$.
The positive errors $e_t$ in the latter case sum to at most
$A$, while the additive terms in both cases sum to at most
$3/8$. Thus,
\begin{equation}
\label{eq:commit-regret-decomposition}
    \overline R\leq C+N_U+F+A+\frac38.
\end{equation}
By Equations~\eqref{eq:other-commit-errors}
and~\eqref{eq:commit-regret-decomposition}, it remains to prove
that $A=\mathcal O(B+F+\beta)$.
The term $B$ appears because matrix insertions on rounds
contributing to $B$ can also reduce the checking probabilities
and delay the detection of later underestimations.

\paragraph{Step 2. Relating $A$ to failed checks}

To bound $A$, we use the fact that the matrix $M$ is scaled by $\alpha^2$ after a reset. We define the potential:
\begin{equation}
\label{eq:matrix-error-weight}
    S_t=\sqrt{{w^\star}^\top M_tw^\star}+2\beta,
    \qquad
    S_1=\eta\|w^\star\|_2+2\beta,
\end{equation}
Its matrix-dependent part measures the size of $M_t$ in the
fixed direction $w^\star$.
Let:
\[
    W=\sum_{t\in[T]}G_tI_tS_t.
\]
Every round counted in $W$ is an unattacked round in which $U_t$ is posted
and $e_t>1/(4T)$. Thus, a sale occurs at $U_t$, revealing the
underestimation and triggering a reset. The quantity $W$
counts these detections, assigning weight $S_t$ to each one.

We first relate the checking probabilities to the errors in $A$.
On a round with $G_t=1$, Cauchy-Schwarz gives:
\[
    e_t
    =|\langle w^\star,a_t\rangle|
\le
\|M_t^{1/2}w^\star\|_2
\|M_t^{-1/2}a_t\|_2
=
\sqrt{(w^\star)^\top M_tw^\star}\,
\sqrt{a_t^\top M_t^{-1}a_t}
\leq\gamma_t\sqrt{{w^\star}^\top M_tw^\star}
    \leq\gamma_tS_t.
\]
If $q_t=\beta\gamma_t$, this immediately implies
$q_tS_t\geq\beta e_t$. If instead $q_t=1/2$, the additive
term $2\beta$ ensures that:
\[
    q_tS_t=\frac{S_t}{2}\geq\beta\geq\beta e_t,
\]
where we used $e_t\leq1$. Hence, in both cases, whenever $G_t=1$, it holds:
\begin{equation}
\label{eq:matrix-error-checking}
    q_tS_t\geq\beta e_t.
\end{equation}
Multiplying by $G_t$ and summing yields:
\[
    \sum_{t\in[T]}G_tq_tS_t\geq\beta A.
\]
Thus, a large value of $A$ forces  $\sum_{t\in[T]}G_tq_tS_t$ to be large.

We next bound $W$ using that $M_t$ is scaled
by $\alpha^2$ after a reset. By
Equations~\eqref{eq:augmented-error}
and~\eqref{eq:matrix-radius-growth}, it holds:
\[
\sqrt{(w^\star)^\top\overline M_tw^\star}+2\beta
\le S_t+|e_t|.
\]
Moreover, every round satisfies:
\[
S_{t+1}\le S_t+|e_t| -(1-\alpha)G_tI_t(S_t-2\beta),
\]
because an insertion adds at most $|e_t|$, and when $G_tI_t=1$ a reset decreases the potential by at least $(1-\alpha)(S_t-2\beta)$. Other resets can only decrease the potential, while non-committing rounds leave it unchanged.
Applying this repeatedly gives, for every $t$,
\begin{equation}
\label{eq:matrix-error-weight-maximum}
S_t
\le S_1+\sum_{s=1}^{t-1}|e_s|
\le S_1+\sum_{s=1}^{T}|e_s| =S_1+A+B
    \leq S_1+T.
\end{equation}
where in  the equality we use the definition of $A$ and $B$.

Similarly, summing over time and using
$\sum_tG_tI_t\leq F$ gives:
\begin{equation}
\label{eq:weighted-failures-telescoping}
    S_{T+1}
    \leq S_1+A+B-(1-\alpha)W
       +2\beta(1-\alpha)F.
\end{equation}
Since $S_{T+1}\geq0$, this bounds:
\begin{equation}
\label{eq:new1}
    (1-\alpha)W
    \leq S_1+A+B+2\beta(1-\alpha)F.
\end{equation}
It remains to connect $W$ to the checking probabilities.
Before the price is drawn, $G_t$, $q_t$, and $S_t$ are already
determined, and $I_t$ has conditional mean $q_t$.
Applying Lemma~\ref{lem:weighted-check-concentration}, and since
$\max_tS_t\leq S_1+T$ by \Cref{eq:matrix-error-weight-maximum}, gives, with probability at least
$1-\delta/2$,
\begin{align}
    W
    &\geq \frac12\sum_{t\in[T]}G_tq_tS_t
       -\left(\max_{t\in[T]}S_t+1\right)\log(2T/\delta)
       \geq \frac{\beta A}{2}
       -(S_1+A+B+1)\log(2T/\delta).
    \label{eq:weighted-checking-probability}
\end{align}
We now have both an upper and a lower bound on the same
quantity $W$. Substituting this lower bound into
Equation~\eqref{eq:new1}, and using
the choices of $\alpha$ and $\beta$, will give
$A=\mathcal O(B+F+\beta)$.

\paragraph{Step 3. Bounding the regret}
The parameter choice in Equation~\eqref{eq:checking-multiplier}
gives $\beta\geq16$, $\beta\geq8\log(2T/\delta)$, and
$\alpha=\exp(-4/\beta)$. The inequalities
$1-x\leq\exp(-x)\leq1-x+x^2/2$ for $x\geq0$ imply:
\[
    \frac72\leq4-\frac8\beta
    \leq\beta(1-\alpha)\leq4,
    \qquad
    (1-\alpha)\log(2T/\delta)\leq\frac12.
\]
We substitute Equation~\eqref{eq:weighted-checking-probability}
into Equation~\eqref{eq:weighted-failures-telescoping}.
Using the previous parameter estimates gives:
\begin{align*}
    0\leq S_{T+1}
    &\leq S_1+A+B+2\beta(1-\alpha)F
        -\frac{\beta(1-\alpha)}2A
        +(1-\alpha)(S_1+A+B+1)\log(2T/\delta)\\
    &\leq S_1+A+B+8F-\frac74A
        +\frac12(S_1+A+B+1)\\
    &=\frac32(S_1+B)+8F+\frac12-\frac14A.
\end{align*}
Moving $A/4$ to the left and multiplying by four yields
$A\leq6(S_1+B)+32F+2$.
Finally, $S_1=2\beta+\eta\|w^\star\|_2=\mathcal{O}(\beta)$
and Equation~\eqref{eq:other-commit-errors} give:
\begin{align*}
    A\leq6(S_1+B)+32F+2
    =\mathcal{O}(B+F+\beta)
    =\mathcal{O}(C+N_U+F+\beta).
\end{align*}
Substituting this estimate into
Equation~\eqref{eq:commit-regret-decomposition} gives
the claimed regret bound:
\[
    \overline R=\mathcal{O}(C+F+N_U+\beta).
\]
The event used above has probability at least
$1-\delta/2$, and all subsequent inequalities hold on that
same event. This proves the claimed probability and concludes
the proof.
\end{proof}

\bibliographystyle{alpha}
\bibliography{example_paper}

\end{document}